\pdfoutput=1
\documentclass[thm-restate,cleveref,english]{lipics-v2021}

\nolinenumbers
\hideLIPIcs

\usepackage{ellipsis, ragged2e} 
\usepackage[l2tabu, orthodox]{nag} 

\usepackage{amsmath,amsthm,amssymb}
\usepackage{csquotes}
\usepackage{booktabs}
\usepackage{thmtools}

\usepackage{framed}
\usepackage{pgfplots}
\pgfplotsset{compat=1.16}
\usepackage{tikz}
\usetikzlibrary{shapes,positioning,arrows,arrows.meta,calc,automata,matrix,fit,backgrounds}
\tikzstyle{state}+=[minimum size = 6mm, inner sep=0,outer sep=1]

\colorlet{disabled}{lightgray}
\tikzstyle{state}=[draw,rectangle,inner sep=5pt,rounded corners=2pt]
\tikzstyle{action}=[font=\small,inner sep=0pt,outer sep=3pt]
\tikzstyle{actionnode}=[circle,draw=black,fill=black,minimum size=1mm,inner sep=0,outer sep=0]
\tikzstyle{actionedge}=[draw,-]
\tikzstyle{prob}=[font=\scriptsize,inner sep=0pt,outer sep=1pt]
\tikzstyle{probedge}=[draw,->]
\tikzstyle{directedge}=[draw,->]
\tikzset{chainarrow/.tip={Stealth[length=3pt]}}
\tikzset{>=chainarrow}

\usepackage{xparse}
\usepackage{mathtools}
\usepackage{environ}
\usepackage{marvosym}
\usepackage{bbm}
\usepackage{fontawesome5}

\usepackage{url}
\usepackage{algorithmicx,algorithm}
\usepackage[noend]{algpseudocode}

\usepackage{booktabs}
\usepackage{multirow}

\newtheorem{assumption}{Assumption}

\DeclarePairedDelimiter{\delimabs}{\lvert}{\rvert}
\DeclarePairedDelimiter{\delimcardinality}{\lvert}{\rvert}
\DeclarePairedDelimiter{\delimnorm}{\lVert}{\rVert}

\NewDocumentCommand{\abs}{sm}{\IfBooleanTF{#1}{\delimabs*{#2}}{\delimabs{#2}}}
\NewDocumentCommand{\cardinality}{sm}{\IfBooleanTF{#1}{\delimcardinality*{#2}}{\delimcardinality{#2}}}
\NewDocumentCommand{\norm}{sm}{\IfBooleanTF{#1}{\delimnorm*{#2}}{\delimnorm{#2}}}
\NewDocumentCommand{\powerset}{r()}{2^{#1}}

\newcommand{\indicator}[1]{\mathbbm{1}_{#1}}

\DeclareMathOperator{\lAnd}{\bigwedge}
\DeclareMathOperator{\lOr}{\bigvee}

\newcommand{\unionSym}{\cup}
\newcommand{\unionBin}{\mathbin{\unionSym}}

\newcommand{\intersectionSym}{\cap}
\newcommand{\intersectionBin}{\mathbin{\intersectionSym}}
\newcommand{\UnionSym}{\bigcup}

\newcommand{\union}{\unionBin}

\newcommand{\intersection}{\intersectionBin}
\newcommand{\Union}{\UnionSym}

\newcommand{\Naturals}{\mathbb{N}}
\newcommand{\Integers}{\mathbb{Z}}
\newcommand{\Rationals}{\mathbb{Q}}
\newcommand{\Reals}{\mathbb{R}}

\newcommand{\RealsNonneg}{\mathbb{R}_{\geq 0}}

\DeclareMathOperator{\support}{supp}

\newcommand{\distribution}{d}

\NewDocumentCommand{\Distributions}{d()}{\IfNoValueTF{#1}{\mathcal{D}}{\mathcal{D}(#1)}}
\NewDocumentCommand{\Measures}{d()}{\IfNoValueTF{#1}{\Pi}{\Pi(#1)}}
\NewDocumentCommand{\integral}{d<> m m}{\IfNoValueTF{#1}{\int #2\,d#3}{\int_{#1} #2\,d#3}}
\NewDocumentCommand{\Expectation}{s d[]}{\IfNoValueTF{#2}{\mathbb{E}}{\mathbb{E}\IfBooleanTF{#1}{\left[#2\right]}{[#2]}}}
\NewDocumentCommand{\Probability}{s d[]}{\mathop{\mathrm{Pr}}\IfValueT{#2}{\IfBooleanTF{#1}{\left[#2\right]}{[#2]}}}

\newcommand{\MC}{\mathsf{M}}
\newcommand{\MDP}{\mathcal{M}}
\newcommand{\SG}{\mathsf{G}}

\newcommand{\States}{S}
\newcommand{\StatesMax}{S_{\max}}
\newcommand{\StatesMin}{S_{\min}}
\newcommand{\initialstate}{\hat{s}}
\newcommand{\Actions}{A}
\NewDocumentCommand{\stateactions}{r()}{{\Actions}(#1)}
\NewDocumentCommand{\mctransitions}{d()}{\IfNoValueTF{#1}{\delta}{\delta(#1)}}
\NewDocumentCommand{\mdptransitions}{d()}{\IfNoValueTF{#1}{\Delta}{\Delta(#1)}}
\NewDocumentCommand{\sgtransitions}{d()}{\IfNoValueTF{#1}{\Delta}{\Delta(#1)}}
\NewDocumentCommand{\reward}{d()}{\IfNoValueTF{#1}{{r}}{{r}(#1)}}

\newcommand{\infinitepath}{\rho}
\newcommand{\finitepath}{\varrho}
\NewDocumentCommand{\Infinitepaths}{d<>}{\IfNoValueTF{#1}{\mathsf{Paths}}{\mathsf{Paths}_{#1}}}
\NewDocumentCommand{\Finitepaths}{d<>}{\IfNoValueTF{#1}{\mathsf{FPaths}}{\mathsf{FPaths}_{#1}}}
\newcommand{\strategy}{\pi}
\newcommand{\stratmax}{\sigma}
\newcommand{\stratmin}{\tau}
\NewDocumentCommand{\Strategies}{d<>}{\IfNoValueTF{#1}{\Pi}{\Pi_{#1}}}
\NewDocumentCommand{\StrategiesM}{d<>}{\IfNoValueTF{#1}{\Pi}{\Pi_{#1}}^{\mathsf{M}}}
\NewDocumentCommand{\StrategiesMD}{d<>}{\IfNoValueTF{#1}{\Pi}{\Pi_{#1}}^{\mathsf{MD}}}

\DeclareMathOperator{\SccsOp}{SCC}
\DeclareMathOperator{\BsccsOp}{BSCC}
\DeclareMathOperator{\EcsOp}{EC}
\DeclareMathOperator{\MecsOp}{MEC}

\NewDocumentCommand{\Sccs}{r()}{\SccsOp(#1)}
\NewDocumentCommand{\Bsccs}{r()}{\BsccsOp(#1)}
\NewDocumentCommand{\Ecs}{d()}{\IfNoValueTF{#1}{\EcsOp}{\EcsOp(#1)}}
\NewDocumentCommand{\Mecs}{d()}{\IfNoValueTF{#1}{\MecsOp}{\MecsOp(#1)}}

\NewDocumentCommand{\ProbabilityMC}{s r<> d[]}{\mathsf{Pr}_{#2}\IfNoValueF{#3}{\IfBooleanTF{#1}{\!\left[#3\right]\!}{[#3]}}}
\NewDocumentCommand{\ProbabilityMDP}{s r<> r<> d[]}{\mathsf{Pr}_{#2}^{#3}\IfNoValueF{#4}{\IfBooleanTF{#1}{\!\left[#4\right]\!}{[#4]}}}
\NewDocumentCommand{\ProbabilitySG}{s r<> r<> d[]}{\mathsf{Pr}_{#2}^{#3}\IfNoValueF{#4}{\IfBooleanTF{#1}{\!\left[#4\right]\!}{[#4]}}}
\NewDocumentCommand{\ProbabilityMDPmax}{s r<> d[]}{\mathsf{Pr}_{#2}^{\max}\IfNoValueF{#3}{\IfBooleanTF{#1}{\!\left[#3\right]\!}{[#3]}}}
\NewDocumentCommand{\ProbabilityMDPsup}{s r<> d[]}{\mathsf{Pr}_{#2}^{\sup}\IfNoValueF{#3}{\IfBooleanTF{#1}{\!\left[#3\right]\!}{[#3]}}}

\NewDocumentCommand{\ExpectationMC}{s r<> d[]}{\mathbb{E}_{#2}\IfValueT{#3}{\IfBooleanTF{#1}{\!\left[#3\right]\!}{[#3]}}}
\NewDocumentCommand{\ExpectationMDP}{s r<> r<> d[]}{\mathbb{E}_{#2}^{#3}\IfValueT{#4}{\IfBooleanTF{#1}{\!\left[#4\right]\!}{[#4]}}}
\NewDocumentCommand{\ExpectationSG} {s r<> r<> d[]}{\mathbb{E}_{#2}^{#3}\IfValueT{#4}{\IfBooleanTF{#1}{\!\left[#4\right]\!}{[#4]}}}

\NewDocumentCommand{\ExpectedSum}{m m}{#1\langle#2\rangle}

\NewDocumentCommand{\ExpectedSumSG}{m m m m}{\ExpectedSum{#1(#2,#3)}{#4}}

\newcommand{\reach}{\lozenge}
\newcommand{\Target}{T}

\DeclareMathOperator{\totalreward}{TR}
\newcommand{\val}{\mathsf{Val}}

\DeclareMathOperator{\Variance}{Var}
\DeclareMathOperator{\opt}{opt}
\DeclareMathOperator{\ERisk}{ERisk}
\DeclareMathOperator{\Utility}{NegUtil}

\NewDocumentCommand{\ERiskSG}{r<> r<> r[]}{\ERisk_{#1}^{#2}(#3)}

\title{Entropic Risk for Turn-Based Stochastic Games}

\author{Christel Baier}{Technische Universit\"at Dresden, Germany}{christel.baier@tu-dresden.de}{0000-0002-5321-9343}{}
\author{Krishnendu Chatterjee}{Institute of Science and Technology Austria (ISTA)}{Krishnendu.Chatterjee@ist.ac.at}{0000-0002-4561-241X}{}
\author{Tobias Meggendorfer}{Institute of Science and Technology Austria (ISTA); Technische Universit\"at M\"unchen, Germany}{tobias@meggendorfer.de}{0000-0002-1712-2165}{}
\author{Jakob Piribauer}{Technische Universit\"at Dresden, Germany; Technische Universit\"at M\"unchen, Germany}{jakob.piribauer@tu-dresden.de}{0000-0003-4829-0476}{}

\authorrunning{C. Baier, K. Chatterjee, T. Meggendorfer, J. Piribauer} 

\Copyright{Christel Baier, Krishnendu Chatterjee, Tobias Meggendorfer, Jakob Piribauer} 



\funding{This work was partly funded by the ERC CoG 863818 (ForM-SMArt), the DFG Grant 389792660 as part of TRR 248 (Foundations of Perspicuous Software Systems), the Cluster of Excellence EXC 2050/1 (CeTI, project ID 390696704, as part of Germany’s Excellence Strategy), and  the DFG projects BA-1679/11-1 and BA-1679/12-1.}

\relatedversion{This is the extended version of a paper accepted for publication at MFCS 2023.}

\usepackage[final]{todonotes}
\setuptodonotes{inline}

\begin{document}

\maketitle

\begin{abstract}
\emph{Entropic risk (ERisk)} is an established risk measure in finance, quantifying risk by an exponential re-weighting of rewards.
We study  ERisk for the first time in the context of turn-based stochastic games with the total reward objective.  This gives rise  to an objective function that demands the control of systems in a risk-averse manner.
We show that the resulting games are determined and, in particular, admit optimal memoryless deterministic strategies.
This contrasts risk measures that previously have been considered in the special case of Markov decision processes and that require randomization and/or memory.
We provide several results on the decidability and the computational complexity of the threshold problem, i.e.\ whether the optimal value of ERisk exceeds a given threshold.
In the most general case, the problem is decidable subject to Shanuel's conjecture.
If all inputs are rational, the resulting threshold problem can be solved using algebraic numbers, leading to decidability via a polynomial-time reduction to the existential theory of the reals. 
Further restrictions on the encoding of the input allow the solution of the threshold problem in $\mathsf{NP}\cap\mathsf{coNP}$.
Finally, an approximation algorithm for the optimal value of ERisk is provided.
\end{abstract}

\section{Introduction}

\noindent
\textbf{Stochastic Models.}
Formal analysis of stochastic models is ubiquitous across disciplines of science, such as computer science \cite{DBLP:books/daglib/0020348}, biology \cite{paulsson2004summing}, epidemiology \cite{G_mez_2010}, and chemistry \cite{gillespie1976general}, to name a few.
In computer science, a fundamental stochastic model are Markov decision processes (MDPs) \cite{DBLP:books/wi/Puterman94}, which extend purely stochastic Markov chains (MCs) with non-determinism to represent an agent interacting with a stochastic environment.
Stochastic games (SGs) \cite{shapley1953stochastic,condon1990algorithms,condon1992complexity} in turn generalize MDPs by introducing an adversary, modelling the case where two agents engage in adversarial interaction in the presence of a stochastic environment.
Notably, SGs can also be used to conservatively model MDPs where transition probabilities are not known precisely \cite{DBLP:conf/fossacs/ChatterjeeSH08,DBLP:conf/cdc/WeiningerMK19}.
See also \cite{white1993survey,white1985real,DBLP:books/wi/Puterman94} and \cite{DBLP:journals/jcss/ChatterjeeH12,filar2012competitive} for further applications of MDPs and SGs.

\noindent
\textbf{Strategies and Objectives.}
In MDPs and SGs, the recipes to resolve choices are called strategies.
The objective of the agent is to optimize a payoff function against all possible strategies of the adversary.
One of the most fundamental problems studied in the context of MDPs and SGs is the optimization of total reward (and the related \emph{stochastic shortest path} problem \cite{DBLP:journals/mor/BertsekasT91}).
Here, every state (or, equivalently, transition) of the stochastic model is assigned a cost or reward and the payoff of a trajectory is the total sum of rewards appearing along the path.
MDPs and SGs with total reward objectives provide an appropriate model to study a wide range of applications, such as traffic optimization \cite{fu1998expected}, verification of stochastic systems \cite{DBLP:conf/sfm/ForejtKNP11,DBLP:conf/vmcai/RandourRS15}, or navigation / probabilistic planning \cite{DBLP:conf/atal/Teichteil-KonigsbuchKI10}.

\noindent
\textbf{Risk-Ignorance of Expectation.}
Typically, the expectation of the obtained total reward is optimized.
However, the expectation measure is ignorant towards aspects of risk; an expectation maximizing agent accepts a one-in-a-million chance of extremely high rewards over a slightly worse, but guaranteed outcome.
Such a behaviour might be undesirable in a lot of situations:
Consider a one-shot lottery where with a chance of $10^{-6}$ we win $2 \cdot 10^6$ times our stake and otherwise lose everything -- a two-times increase in expectation.
The optimal strategy w.r.t.\ expectation would bet all available assets, ending up broke in nearly all outcomes.

\noindent
\textbf{Risk-Aware Alternatives.}
To address this issue, \emph{risk-aware} objectives create incentives to prefer slightly smaller performance in terms of expectation in exchange for a more \enquote{stable} behaviour.
To this end, several variants have been studied in the verification literature, such as
	(a)~variance-penalized expected payoff~\cite{DBLP:conf/icalp/PiribauerSB22,DBLP:journals/mor/FilarKL89} that combines the expected value with a penalty for the variance of the resulting probability distribution;
	(b)~trade-off of the expectation and variance for various notions of variance~\cite{DBLP:conf/icml/MannorT11,DBLP:journals/jcss/BrazdilCFK17};
	(c)~quantiles and conditional value-at-risk (CVaR)~\cite{DBLP:conf/vmcai/RandourRS15,DBLP:conf/aaai/Meggendorfer22,DBLP:conf/lics/KretinskyM18}; to name a few.

\noindent
\textbf{Drawbacks.}
The current approaches suffer from the following three drawbacks:
\begin{enumerate}
	\item
	The above studies focus on the second moment (variance) along with the first moment (mean), but do not incorporate other moments of the payoff distribution.

	\item
	All approaches are studied only for MDPs; none of them have been extended to SGs.

	\item
	Even in MDPs,  the above problems require complicated strategies.
	For example,  trade-offs between expectation and variance require memory and randomization \cite{DBLP:journals/jcss/BrazdilCFK17,DBLP:conf/icml/MannorT11}, while optimizing variance-penalized expected payoffs, quantiles, or the CVaR of the total reward require exponential memory \cite{DBLP:conf/icalp/PiribauerSB22,DBLP:conf/icalp/HaaseK15,DBLP:conf/icalp/PiribauerB20,DBLP:conf/aaai/Meggendorfer22}.
\end{enumerate}

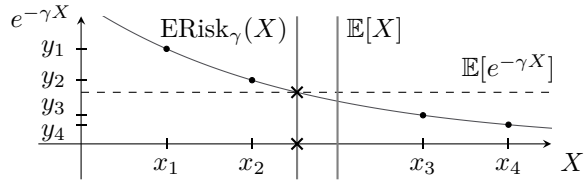
\begin{figure}[t]
	\centering
	\pgfplotsset{ticks=none}
	\begin{tikzpicture}[auto]
		\tikzset{
			xtick/.style={fill,rectangle,minimum width=0.75pt,minimum height=4pt,inner sep=0pt},
			ytick/.style={fill,rectangle,minimum width=4pt,minimum height=0.75pt,inner sep=0pt}
		}
		\begin{axis}[
			axis lines = middle,
			xlabel={$X$}, x label style={anchor=north west},
			ylabel={$e^{-\gamma X}$}, y label style={anchor=east},
			ymin=-.25, ymax=0.9,
			xmin=-.5, xmax=5.5,
			height=3.75cm,width=0.6\columnwidth
		]
			\addplot [domain=0:6, samples=300, color=darkgray, smooth] {exp(1)^(-.4*x)};
			\addplot [domain=0.6, color=black, samples at={1,2,4,5}, only marks,mark=*,mark size=1pt] expression {exp(1)^(-.4*x)};

			\draw[dashed] ({axis cs:0,0.364}) -- ({axis cs:6,0.364});
			\draw[solid,gray,thick] ({axis cs:3,0}|-{rel axis cs:0,0}) -- ({axis cs:3,0}|-{rel axis cs:0,1});
			\draw[solid,gray,thick] ({axis cs:2.5265,0}|-{rel axis cs:0,0}) -- ({axis cs:2.5265,0}|-{rel axis cs:0,1});

			\addplot [only marks,mark=x,mark size=3pt,thick] coordinates {(2.5265, 0) (2.5265, {exp(1)^(-.4*2.5265)})};

			\node[label={-90:{$x_1$}},xtick] at (axis cs:1,0) {};
			\node[label={180:{$y_1$}},ytick] at (axis cs:0,{exp(1)^(-.4*1)}) {};
			
			\node[label={-90:{$x_2$}},xtick] at (axis cs:2,0) {};
			\node[label={180:{$y_2$}},ytick] at (axis cs:0,{exp(1)^(-.4*2)}) {};
			
			\node[label={-90:{$x_3$}},xtick] at (axis cs:4,0) {};
			\node[label={[anchor=south east,outer ysep=0pt,inner ysep=0pt]180:{$y_3$}},ytick] at (axis cs:0,{exp(1)^(-.4*4)}) {};

			\node[label={-90:{$x_4$}},xtick] at (axis cs:5,0) {};
			\node[label={[anchor=north east,outer ysep=0pt,inner ysep=0pt]180:{$y_4$}},ytick] at (axis cs:0,{exp(1)^(-.4*5)}) {};

			\node[anchor=west] at (axis cs:3,0.8)   (E) {$\Expectation[X]$};
			\node[anchor=south] at (axis cs:5,.364)   (E) {$\Expectation[e^{-\gamma X}]$};
			\node[anchor=east] at (axis cs:2.5265,0.8) (F) {$\ERisk_{\gamma}(X)$};
		\end{axis}
	\end{tikzpicture}
	\caption{
		Illustration of the entropic risk measure.
		The random variable $X$ takes values $x_1$ to $x_4$ uniformly with probability $\tfrac{1}{4}$ each.
		Expectation considers the average of $x_i$, while entropic risk yields the (normalized logarithm of the) average of $y_i = e^{-\gamma x_i}$.
	}
	\label{fig:illustration_ERisk}
\end{figure}

\noindent
\textbf{Entropic Risk.}
The notion of entropic risk \cite{follmer2002convex} has been widely studied in finance and operation research, see e.g.~\cite{follmer2011entropic,brandtner2018entropic}.
Informally, instead of weighing each outcome uniformly and then aggregating it (as in the case for regular expectation), entropic risk re-weighs outcomes by an exponential function, then computes the expectation, and finally re-normalizes the value.
We illustrate this in \cref{fig:illustration_ERisk}.
The exact definition of entropic risk is introduced later on.

\noindent
\textbf{Advantages.}
Aside from satisfying many desirable properties of risk measures established in finance, entropic risk brings several crucial advantages in our specific setting, of which we list a few:
Compared to expectation, \enquote{bad} outcomes are penalized more than \enquote{good} outcomes add value.
Thus, an agent optimizing entropic risk seeks to reduce the chances of particularly bad outcomes while  also being interested in a good overall performance.
In contrast to variance minimization, it  is beneficial to increase the probability of extremely good outcomes (which would increase variance).
Moreover, the entropic risk  incorporates \emph{all} moments of the distribution.
In particular, even if the expectation is infinite, entropic risk still provides meaningful values (opposed to both expectation and variance).
Note that the expected total reward objective is often addressed under additional assumptions excluding this case \cite{DBLP:books/lib/BertsekasT96,DBLP:conf/sfm/ForejtKNP11}.
Additionally, entropic risk is a \emph{time-consistent} risk measure.
In our situation, this means that the risk evaluation at a state is the same  for \emph{any history}.
This is in stark contrast to, e.g., quantile and CVaR optimal strategies, which after a series of unfortunate events start behaving recklessly (e.g.\ expectation optimal).
Due to these advantages, ERisk has already been studied in the context of MDPs \cite{howard1972risk,bauerle2014more}.
However, to the best of our knowledge, neither the arising computational problems nor the more general setting of SGs have been addressed.

\subsection{Our results}
In this work we consider the notion of entropic risk in the context of SGs as well as the special cases of MCs and MDPs. For an overview of our complexity results, see \cref{tbl:overview}.
\begin{enumerate}
	\item \emph{Determinacy and Strategy Complexity.}
	We establish several basic results, in particular that SGs with the entropic risk objective are determined and that pure memoryless optimal strategies exist for both players.
	This stands in contrast to other notions of risk, where even in MDPs strategies require memory and/or randomization.

	\item \emph{Exact Computation.}
	When allowing Euler's number $e$ as the basis of exponentiation, the threshold problem whether the optimal entropic risk lies above a given bound is decidable subject to Shanuel's conjecture.
	If the basis of exponentiation and all other numbers in the input are rational, then all numbers resulting from the involved exponentiation are shown to be algebraic.
	We obtain a reduction to the existential theory of the reals and thus a PSPACE upper bound in this case.
	
	Furthermore, we identify a notion of \emph{small algebraic instance} in which all occurring numbers are not only algebraic, but have a  small representation and are contained in an algebraic extensions of $\mathbb{Q}$ of low degree.
	The threshold problem for small algebraic instances of MCs and MDPs can efficiently be solved by explicit computations in an algebraic  extension of $\mathbb{Q}$.
	We obtain polynomial-time algorithms for MCs and MDPs, and conclude that the threshold problem lies in NP $\intersection$ co-NP for SGs in this case.
	For small algebraic instances, we furthermore show that an explicit closed form of the optimal value can be computed (a)  in polynomial time  for MCs; and consequently (b)  in polynomial space for SGs.

	\item \emph{Approximate Computation.}
	We provide an effective way to compute an approximation, i.e.\ determine the optimal entropic risk up to a given precision of $\varepsilon > 0$.
	To this end, we show that in the general case, by considering enough bits of arising irrational numbers, we can bound the incurred error.
	In MDPs and MCs, the optimal value can be approximated in time polynomial in the size of the model, in $-\log(\varepsilon)$, and in the magnitude of the rewards. For SGs, this implies the existence of a polynomial-space approximation algorithm.
\end{enumerate}

\begin{table*}
	\caption{Overview of the decidability and complexity results for SGs, MDPs and MCs.}
	\label{tbl:overview}
	\scalebox{0.84}{
	\begin{tabular}{c | c | c | c | c |   c }
		\toprule
		                   \multicolumn{1}{c}{}                    &                                                                                                               \multicolumn{3}{c}{threshold problem}                                                                                                                &                                               \multicolumn{2}{c}{optimal value}                                                \\
		\cmidrule(lr){2-4} \cmidrule(lr){5-6}
	\multicolumn{1}{c}{} & \multicolumn{1}{c}{ general }                                                                 & \multicolumn{1}{c}{algebraic}                                                                          & \multicolumn{1}{c}{ small algebraic}                      & \multicolumn{1}{c}{computation for}                       & \multicolumn{1}{c}{approximation with small } \\
		                  \multicolumn{1}{c}{ }                    & \multicolumn{1}{c}{ instances  }                                                              & \multicolumn{1}{c}{  instances  }                                                                      & \multicolumn{1}{c}{   instances   }                       & \multicolumn{1}{c}{ small algebraic  }           & \multicolumn{1}{c}{ rewards and risk  }        \\
		                  \multicolumn{1}{c}{ }                    & \multicolumn{1}{c}{ (Thm. \ref{thm:schanuel})}                                                    & \multicolumn{1}{c}{ (Thm. \ref{thm:generalalgebraic}) }                                                    & \multicolumn{1}{c}{ (Thm. \ref{thm:smallalgebraicthreshold}) } & \multicolumn{1}{c}{instances (Thm. \ref{thm:algebraic_optimal_value})} & \multicolumn{1}{c}{aversion factor (Thm. \ref{thm:approximation})}                    \\
		\midrule
		                  \rule{0pt}{2.6ex}
	
	SGs                   & \multirow{3}{2cm}[0.5ex]{ \centering decidable subject to Shanuel's conjecture} & \multirow{3}{2cm}[-0.3ex]{\centering in PSPACE (in $\exists \Reals$)} & in NP $\cap$ {coNP}                                       & \multirow{2}{3cm}[0pt]{\centering in polynomial space}         & in polynomial space                                                \\
		\cline{4-4}\cline{6-6}
		                 \rule{0pt}{2.6ex}

	MDPs                  &                                                                                               &                                                                                                        & \multirow{2}{2cm}[-2pt]{\centering in PTIME}                    &                                                           & \multirow{2}{3cm}[-2pt]{\centering in polynomial time}                   \\
		\cline{5-5}
		                  \rule{0pt}{2.6ex}
	
	MCs                   &                                                                                               &                                                                                                        &                                                           & in polynomial time                                        &                                                                    \\
		\bottomrule
	\end{tabular}
	}
\end{table*}

\subsection{Related Work}
The entropic risk objective has been studied before in MDPs:
An early formulation  can be found in \cite{howard1972risk} under the name \emph{risk-sensitive MDPs} focusing on the finite-horizon setting.
The paper \cite{jaquette1976utility} considers an exponential utility function applied to discounted rewards and optimal strategies are shown to exist, but not to be memoryless in general.
In \cite{di1999risk}, the entropic risk objective is considered for MDPs with a general Borel state space and in \cite{bauerle2014more} a generalization of this objective is studied on such MDPs.
To the best of our knowledge, however, all previous work in the context of MDPs focuses on optimality equations and general convergence results of value iteration, while the resulting algorithmic problems for finite-state MDPs have not been investigated.
Furthermore, we are not aware of work on the entropic risk objective in SGs.

For other objectives capturing  risk-aversion, algorithmic problems have been analyzed on finite-state MDPs:
Variance-penalized expectation has been studied for finite-horizon MDPs with terminal rewards in \cite{collins1997finite} and for infinite-horizon MDPs with discounted rewards and mean payoffs \cite{DBLP:journals/mor/FilarKL89},  and  total rewards \cite{DBLP:conf/icalp/PiribauerSB22}.
For total rewards,  optimal strategies require exponential memory and  the threshold problem is in NEXPTIME and  EXPTIME-hard  \cite{DBLP:conf/icalp/PiribauerSB22}.

In \cite{DBLP:conf/icml/MannorT11}, the optimization of expected accumulated rewards under constraints on the variance are studied for finite-horizon MDPs.
Possible tradeoffs between expected value and variance of mean payoffs  and other notions of variability have been studied in~\cite{DBLP:journals/jcss/BrazdilCFK17}.

To control the chance of bad outcomes, the problem to maximize or minimize the probability that the accumulated weight lies below a given bound $w$ has been addressed in MDPs \cite{DBLP:conf/icalp/HaaseK15,DBLP:conf/lics/HaaseKL17}.
Similarly, quantile queries ask for the minimal weight $w$ such that the weight of a path stays below $w$ with probability at least $p$ for the given value $p$ under some or all schedulers \cite{DBLP:conf/fossacs/UmmelsB13,DBLP:journals/fmsd/RandourRS17}.
Both of these problems have been addressed for MDPs with non-negative weights and are solvable in exponential time in this setting \cite{DBLP:conf/fossacs/UmmelsB13,DBLP:conf/icalp/HaaseK15}.
Optimal strategies require exponential memory and the decision version of these problems is PSPACE-hard \cite{DBLP:conf/icalp/HaaseK15}.

The conditional value-at-risk (CVaR), a prominent risk-measure, has been investigated for mean payoff and weighted reachability in MDPs in \cite{DBLP:conf/lics/KretinskyM18} as well as for total rewards in MDPs \cite{DBLP:conf/icalp/PiribauerB20,DBLP:conf/aaai/Meggendorfer22}.
The optimal CVaR of the total reward in MDPs with non-negative weights can be computed in exponential time and optimal strategies require exponential memory \cite{DBLP:conf/icalp/PiribauerB20,DBLP:conf/aaai/Meggendorfer22}.
The threshold problem for optimal CVaR of total reward in MDPs with integer weights is at least as hard as the Positivity-problem for linear recurrence sequences, a well-known problem in analytic number theory whose decidability status is,  since many decades, open \cite{DBLP:conf/icalp/PiribauerB20}.

For all these objectives capturing risk-aversion in some sense, we are not aware of any work addressing the resulting algorithmic problems on SGs.

\section{Preliminaries}

In this section, we  recall the basics of (turn-based) SGs and relevant objectives.
For further details, see, e.g., \cite{DBLP:books/wi/Puterman94,DBLP:books/daglib/0020348,DBLP:conf/sfm/ForejtKNP11,filar2012competitive}.
We assume familiarity with basic notions of probability theory (see, e.g., \cite{billingsley2008probability}).
We write $\Distributions(X)$ to denote the set of all \emph{probability distributions} over a countable set $X$, i.e.\ mappings $\distribution : X \to [0, 1]$ such that $\sum_{x \in X} \distribution(x) = 1$.
The support of a distribution $\distribution$ is  $\support(\distribution) \coloneqq \{x \in X \mid \distribution(x) > 0\}$.
For a set $S$, $S^\star$ and $S^\omega$ refer to the set of finite and infinite sequences of elements of $S$, respectively.

\paragraph*{Markov Chains, MDPs, and Stochastic Games}
	A \emph{Markov chain (MC)} (e.g.\ \cite{DBLP:books/daglib/0020348}), is a tuple $\MC = (\States, \mctransitions)$, where
	$\States$ is a set of \emph{states}, and
	$\mctransitions : \States \to \Distributions(\States)$ is a \emph{transition function} that for each state $s$ yields a probability distribution over successor states.
	We write $\mctransitions(s,s^\prime)$ instead of $\mctransitions(s)(s^\prime)$ for the probability to move from $s$ to $s^\prime$ for $s,s^\prime\in S$.
	A \emph{(infinite) path} in an MC is an infinite sequence $s_0,s_1, \dots$ of states such that for all $i$, we have $\mctransitions(s_i,s_{i+1})>0$.
	We denote the set of infinite paths by $\Infinitepaths<\MC>$.
Together with a state $s$, an MC $\MC$ induces a unique probability distribution $\ProbabilityMC<\MC, s>$ over the set of all infinite paths $\Infinitepaths<\MC>$ starting in $s$.
For a random variable  $f : \Infinitepaths<\MC> \to \Reals$, we write $\ExpectationMC<\MC, s>(f)$ for the expected value of $f$ under the probability measure $\ProbabilityMC<\MC, s>$.

	A \emph{turn-based stochastic game (SG)} (e.g.~\cite{condon1990algorithms}) is a tuple $(\StatesMax, \StatesMin, \Actions, \sgtransitions)$, where
	$\StatesMax$ and $\StatesMin$ are disjoint sets of \emph{Maximizer} and \emph{Minimizer} states,
	inducing the set of states $\States = \StatesMax \union \StatesMin$,
	$\Actions$ denotes a finite set of \emph{actions}, furthermore overloading $\Actions$ to also act as a function assigning to each state $s$ a set of non-empty \emph{available actions} $\Actions(s) \subseteq \Actions$, and
	$\sgtransitions : \States \times \Actions \to \Distributions(\States)$ is the \emph{transition function} that for each state $s$ and (available) action $a \in \Actions(s)$ yields a distribution over successor states.
For convenience, we write $\sgtransitions(s, a, s')$ instead of $\sgtransitions(s, a)(s')$.
Moreover, $\opt_{a \in \stateactions(s)}^s$ refers to $\max_{a \in \stateactions(s)}$ if $s\in \StatesMax$ and $\min_{a \in \stateactions(s)}$ if $s\in \StatesMin$, i.e.\ the preference of either player in a state $s$.
We omit the superscript $s$ where clear from context.
Given a function $f : \States \to \Reals$ assigning values to states, we write $\ExpectedSumSG{\sgtransitions}{s}{a}{f} \coloneqq \sum_{s' \in \States} \sgtransitions(s, a, s') \cdot f(s')$ for the weighted sum over the successors of $s$ under $a \in \Actions(s)$.
A \emph{Markov decision process (MDP)} (e.g.\ \cite{DBLP:books/wi/Puterman94}) can  be seen as an SG with only one player, i.e.\ $\StatesMax = \emptyset$ or $\StatesMin = \emptyset$.

The semantics of SGs is given in terms of resolving choices by strategies inducing an MC with the respective probability space over infinite paths.
Intuitively, a stochastic game is played in turns:
In every state $s$, the player to whom it belongs chooses an action $a$ from the set of available actions $\Actions(s)$ and the play advances to a successor state $s'$ according to the probability distribution given by $\sgtransitions(s, a)$.
Starting in a state $s_0$ and repeating this process indefinitely yields an infinite sequence $\infinitepath = s_0 a_0 s_1 a_1 \dots \in (\States \times \Actions)^\omega$ such that for every $i \in \Naturals_0$ we have $a_i \in \Actions(s_i)$ and $\sgtransitions(s_i, a_i, s_{i+1}) > 0$.
We refer to such sequences as \emph{(infinite) paths} or \emph{plays} and  denote the set of all  infinite paths in a given game $\SG$ by $\Infinitepaths<\SG>$.
Furthermore, we write $\infinitepath_i$ to denote the $i$-th state  in the  path $\infinitepath$.
\emph{Finite paths} or \emph{histories} $\Finitepaths<\SG>$ are finite prefixes of a play, i.e.\ elements of $(\States \times \Actions)^\star \times \States$ consistent with $\Actions$ and $\sgtransitions$.

The decision-making of the players is captured by the notion of \emph{strategies}.
Strategies are functions mapping a given history to a distribution over the actions available in the current state.
For this paper, \emph{memoryless deterministic} strategies (abbreviated \emph{MD strategies}, also called positional strategies) are of particular interest.
These strategies choose a single action in each state, irrespective of the history, and can be identified with functions $\stratmax : \States \to \Actions$.
Since we show that these strategies are sufficient for the discussed notions, we define the semantics of games only for these strategies and refer the interested reader to the mentioned literature for further details.
We write $\Strategies<\SG>$ for the set of all strategies and $\StrategiesMD<\SG>$ for memoryless deterministic ones.
We call a pair of strategies a \emph{strategy profile}, written $\strategy = (\stratmax, \stratmin)$.
We identify a profile with the induced joint strategy $\strategy(s) \coloneqq \stratmax(s)$ if $s \in \StatesMax$ and $\stratmin(s)$ otherwise.

Given a profile $\strategy = (\stratmax, \stratmin)$ of MD strategies for a game $\SG$, we write $\SG^{\strategy}$ for the MC obtained by fixing both strategies.
So, $\SG^{\strategy} = (\States, \hat{\mctransitions})$, where $\hat{\mctransitions}(s) \coloneqq \sgtransitions(s, \strategy(s))$.
Together with a state $s$, the MC $\SG^{\strategy}$ induces a unique probability distribution $\ProbabilitySG<\SG, s><\strategy>$ over the set of all infinite paths $\Infinitepaths<\SG>$.
For a random variable over paths $f : \Infinitepaths<\SG> \to \Reals$, we write $\ExpectationSG<\SG, s><\strategy>[f]$ for the expected value of $f$ under the probability measure $\ProbabilitySG<\SG, s><\strategy>$.

\paragraph*{Objectives}
Usually, we are interested in finding strategies that optimize the value obtained for a particular \emph{objective}.
We introduce some objectives of interest.

\noindent
\textbf{Reachability.} %
A reachability objective is specified by a set of \emph{target states} $\Target \subseteq \States$.
We define $\reach \Target = \{\infinitepath \mid \exists i. \infinitepath_i \in \Target\}$ the set of all paths eventually reaching a target state.
Given a strategy profile $\strategy$ and a state $s$, the probability for this event is given by $\ProbabilitySG<\SG, s><\strategy>[\reach \Target]$.
On games, we are interested in determining the \emph{value} 
$
	\val_{\SG, \reach \Target}(s) \coloneqq {\max}_{\stratmax \in \StrategiesMD<\SG>} {\min}_{\stratmin \in \StrategiesMD<\SG>} \ProbabilitySG<\SG, s><\stratmax, \stratmin>[\reach \Target]
$
of a state $s$, which intuitively is the best probability we can ensure against an optimal opponent.
Generally, one would consider supremum and infimum over strategies instead maximum and minimum over MD strategies.
However, for reachability we know that these value coincide and the game is \emph{determined}, i.e.\ the order of max and min does not matter \cite{condon1992complexity}.
Finally, we know that the value $\val_{\SG, \reach \Target}$ is a solution of the following set of equations
\begin{equation} \label{eq:reachability_equation}
			v(s) = 0 \text{ for $s \in S_0$,} \quad v(s) = 1 \text{ for $s \in \Target$,} \quad \text{and } v(s) = \opt_{a \in \stateactions(s)} \ExpectedSumSG{\sgtransitions}{s}{a}{v} \text{ otherwise},
\end{equation}
where $S_0$ is the set of states that cannot reach $T$  against an optimal Minimizer strategy \cite{DBLP:conf/spin/ChatterjeeH08}.

\noindent
\textbf{Total Reward.} %
The total reward objective
is specified by a reward function $\reward : \States \to \RealsNonneg$, assigning non-negative rewards to every state.
The total reward obtained by a particular path is defined as the sum of all rewards seen along this path, $\totalreward(\infinitepath) \coloneqq \sum_{i=1}^\infty \reward(\infinitepath_i)$.
Note that since we assume $\reward(s) \geq 0$, this sum is always well-defined.
Classically, we want to optimize the expected total reward, i.e.\ determine
$	\val_{\SG, \Expectation\totalreward}(s) \coloneqq {\max}_{\stratmax \in \StrategiesMD<\SG>} {\min}_{\stratmin \in \StrategiesMD<\SG>} \ExpectationSG<\SG, s><\stratmax, \stratmin>[\totalreward]$.
This game is determined and MD strategies suffice \cite{DBLP:journals/fmsd/ChenFKPS13}.
(To be precise, that work considers a more general formulation of total reward, our case is equivalent to the case $\star = c$ and $T = \emptyset$ (Def.~3) and the quantitative rPATL formula $\langle\langle\{1\}\rangle\rangle \textbf{R}^{\reward}_{{\max} = ?}[\textbf{F}^c \texttt{ff}]$.)

\section{Entropic Risk}
As hinted in the introduction, for classical total reward we optimize the expectation and disregard other properties of the actual distribution of obtained rewards.
This means that an optimal strategy may accept arbitrary risks if they yield minimal improvements in terms of expectation.
To overcome this downside, we consider the entropic risk:
\begin{definition}
	Let $b > 1$ a basis, $X$ a random variable, and $\gamma > 0$ a risk aversion factor.
	The \emph{entropic risk  (of $X$ with base $b$ and factor $\gamma$)} (see, e.g., \cite{follmer2004stochastic}) is defined as
	\begin{equation*}
		\ERisk_\gamma(X) \coloneqq - \tfrac{1}{\gamma} \log_b(\Expectation[b^{-\gamma X}]).
	\end{equation*}
\end{definition}
One often chooses $b = e$.
Nevertheless, we also consider rational values for $b$, which allows us to apply techniques from algebraic number theory to  arising computational problems.

\begin{example}
	Consider a random variable $X$ that takes values $x_1=1$, $x_2=2$, $x_3=4$, and $x_4=5$ with probability $1/4$ each.
	\Cref{fig:illustration_ERisk} illustrates how the entropic risk measure of $X$ with base $e$ is obtained for some risk aversion factor $\gamma$:
	The values $x_i$ are depicted on the $x$-axis.
	We now map the values $x_i$ to values $y_i = e^{-\gamma x_i}$ on the $y$-axis.
	Then, the expected value of $e^{-\gamma X}$ can be obtained as the arithmetic mean of the values $y_i$.
	The result is mapped back to the $x$-axis via $y \mapsto - \frac{1}{\gamma} \log(y)$, the inverse of $x \mapsto e^{-\gamma x}$, and we obtain $\ERisk_{\gamma}(X)$.
\end{example}
The example shows that deviations to lower values are penalized, i.e.\ taken into consideration more strongly, by this risk measure.
For a different perspective, we can also consider the Taylor expansion of $\ERisk$ w.r.t.\ $\gamma$, which is
$\ERisk_\gamma(X) = \Expectation[X] - \frac{\gamma}{2} \cdot \Variance[X] + \mathcal{O}(\gamma^2)$ (see, e.g., \cite{asienkiewicz2017note}).
The terms hidden in $\mathcal{O}(\gamma^2)$ comprise all moments of $X$ and exhibit an asymmetry such that $\ERisk$ is roughly the expected value minus a penalty for deviations to lower values.

\subsection{Entropic Risk in SGs}
We are interested in the case $X = \totalreward$, i.e.\ optimizing the risk for total rewards.
We write 
\begin{equation*}
	\ERiskSG<\SG, \initialstate><\gamma>[\strategy] \coloneqq -\tfrac{1}{\gamma} \log_b(\ExpectationSG<\SG, \initialstate><\strategy>[b^{-\gamma X}])
\end{equation*}
to denote the entropic risk of the total reward achieved by the strategy profile $\strategy$ when starting in state $\initialstate$, omitting sub- and superscripts where clear from context.
Clearly, this is well defined for any profile:
We have that $b^{-\gamma \totalreward(\infinitepath)} = b^{- \gamma \sum_{i=1}^\infty \reward(\infinitepath_i)} = \prod_{i=1}^\infty b^{-\gamma \reward(\infinitepath_i)}$ and each factor lies between $0$ and $1$, thus the product converges (possibly with limit $0$).

We also give an insightful characterization for integer rewards.
If $\reward(s) \in \Naturals$, we have
\begin{equation} \label{eq:split_by_total_reward}
	\ERiskSG<\SG, \initialstate><\gamma>[\strategy] = - \tfrac{1}{\gamma} \log_b \left({\sum}_{n=0}^{\infty} \ProbabilitySG<\SG, \initialstate><\strategy>[\totalreward = n] \cdot b^{-\gamma n} \right).
\end{equation}
Naturally, our goal is to optimize the entropic risk.
In this work, we mainly consider the corresponding decision variant, which we call the \emph{entropic risk threshold problem}:
\begin{framed}
	\noindent{\textbf{Entropic risk threshold problem:}}
	Given an SG $\SG$, state $\initialstate$, reward function $\reward$, risk parameter $\gamma$, risk basis $b$, and threshold $t$, decide whether there exists a Maximizer strategy $\stratmax$ such that for all Minimizer strategies $\stratmin$ we have $\ERiskSG<\SG, \initialstate><\gamma>[(\stratmax, \stratmin)] \geq t$.
\end{framed}
Note that (for now) we do not assume any particular encoding of the input.
For example, the reward function $\reward$ could be given symbolically, describing irrational numbers.
A second variant of the threshold problem asks whether the optimal value
\begin{align}
	\ERisk_{\SG, s}^{\gamma *} & \coloneqq {\sup}_{\stratmax \in \Strategies<\SG>} {\inf}_{\stratmin \in \Strategies<\SG>} \ERiskSG<\SG, s><\gamma>[(\stratmax, \stratmin)] \label{eqn:erisk_opt} 
\end{align}
is at least $t$ for a given threshold $t$.
We will see that SGs with the entropic risk as objective function are determined and hence the two variants are equivalent.
Before proceeding with our solution approaches, we provide an illustrative example.
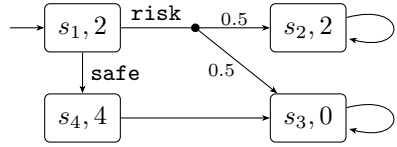
\begin{figure}
	\centering
	\begin{tikzpicture}[auto,initial text=]
		\node[state,initial=left] at (1, 0) (s1) {$s_1, 2$};
		\node[actionnode] at (2.5, 0) (a11) {};
		\node[state] at (4, 0) (s2) {$s_2, 2$};
		\node[state] at (4, -1.2) (s3) {$s_3, 0$};

		\node[state] at (1, -1.2) (s4) {$s_4, 4$};

		\path[->,directedge]
			(s2) edge[loop right] (s2)
			(s3) edge[loop right] (s3)
			(s1) edge node[action] {$\mathtt{safe}$} (s4)
			(s4) edge (s3)
		;
		\path[->,probedge]
			(a11) edge node[prob] {0.5} (s2)
			(a11) edge node[swap,prob] {0.5} (s3)
		;
		\path[-,actionedge]
			(s1) edge node[action] {$\mathtt{risk}$} (a11)
		;
	\end{tikzpicture}
	\caption{
		Our running example to demonstrate several properties of entropic risk.
		For ease of presentation, the system actually is an MDP, where all states belong to Maximizer.
		States are denoted by boxes and their reward is written next to the state name.
		Transition probabilities are written next to the corresponding edges, omitting probability 1.
	}
	\label{fig:running_example}
\end{figure}
\begin{example}
	Consider the MDP of \cref{fig:running_example}.
	The optimal total reward is obtained by choosing action $\mathtt{risk}$ in state $s_1$:
	Then, we actually obtain an infinite total reward through state $s_2$.
	In comparison, choosing action $\mathtt{safe}$ would yield a reward of $6$ in total.
	Now, consider the entropic risk.
	When choosing action $\mathtt{risk}$, we obtain a total reward of $2$ and $\infty$ with probability $\frac{1}{2}$ each, while action $\mathtt{safe}$ yields 6 with probability 1.
	Let $b = 2$ and $\gamma = 1$ for simplicity.
	Then, we obtain an entropic risk of $- \log_2(\frac{1}{2} 2^{-2} + \frac{1}{2} 2^{-\infty}) = 3$ under action $\mathtt{risk}$ and $- \log_2(2^{-6}) = 6$ for $\mathtt{safe}$.
	Thus, action $\mathtt{safe}$ is preferable.
\end{example}
\begin{remark}
	As hinted above, entropic risk is finite whenever a finite reward is obtained with non-zero probability, i.e.\ for any strategy profile $\strategy$, $\ERiskSG<\SG, \initialstate><\gamma>[\strategy] = \infty$ iff $\Probability_{\SG, \initialstate}^{\strategy}[\totalreward = \infty] = 1$.
	In contrast, expectation is infinite whenever there is a non-zero chance of infinite reward, i.e.\ $\Expectation_{\SG, \initialstate}^{\strategy}[\totalreward] = \infty$ iff $\Probability_{\SG, \initialstate}^{\strategy}[\totalreward = \infty]  > 0$.
	So, entropic risk allows us to meaningfully compare strategies which  yield infinite total reward with some positive probability.
\end{remark}
\subsection{Exponential Utility}
Observe that the essential part of the entropic risk is the inner expectation.
Thus, we consider the \emph{negative exponential utility}
\begin{equation*}
	\Utility_{\SG,\initialstate}^{\gamma}(\strategy) \coloneqq \ExpectationSG<\SG, \initialstate><\strategy>[b^{-\gamma \totalreward}].
\end{equation*}
We have $\ERiskSG<\SG, \initialstate><\gamma>[\strategy] = -\frac{1}{\gamma} \log_b(\Utility_{\SG,\initialstate}^{\gamma}(\strategy))$. %
Observe that in our case $0 \leq \Utility_{\SG,\initialstate}^{\gamma}(\strategy) \leq 1$ for any $\strategy$, as $0 \leq \totalreward \leq \infty$.
Moreover, $\ERiskSG<\SG, \initialstate><\gamma>[\strategy] \geq t$ iff $\Utility_{\SG,\initialstate}^{\gamma}(\strategy) \leq b^{- \gamma \cdot t}$, thus, a risk-averse agent (in our case Maximizer) wants to minimize $\Utility$.
The optimal value is
\begin{align}
	\Utility_{\SG, s}^{\gamma *} & \coloneqq {\inf}_{\stratmax \in \Strategies<\SG>} {\sup}_{\stratmin \in \Strategies<\SG>} \ExpectationSG<\SG, s><\stratmax, \stratmin>[b^{-\gamma \totalreward}]. \label{eqn:utility_opt}
\end{align}
We again omit sub- and superscripts where clear from context.
We show later that games with $\Utility$ or $\ERisk$ as payoff functions are determined.
Thus, the order of $\sup$ and $\inf$ in the above definition does not matter.
We call a Maximizer-strategy $\stratmax$ optimal if $\ERisk_{\SG, s}^{\gamma *} = {\inf}_{\stratmin \in \Strategies<\SG>} \ERiskSG<\SG, s><\gamma>[(\stratmax, \stratmin)]$ and analogously for Minimizer-strategies.

\section{Basic Properties and Decidability} \label{sec:basic_porperties}
In this section, we establish several results for SGs with entropic risk as objective functions concerning determinacy, strategy complexity, and decidability in the general case.
We mainly work on games with $\Utility$ as payoff function. As $\ERisk$ can be obtained from $\Utility$ via the monotone function $-\frac{1}{\gamma} \log(\cdot)$,  most results, such as determinacy or strategy complexity, will transfer directly to games with $\ERisk$ as objective function.

First, we show that the games are determined, i.e.\ the order of $\sup$ and $\inf$ in \cref{eqn:erisk_opt} and \cref{eqn:utility_opt} can be switched.
Then, we show that games with $\Utility$ as payoff function can be seen as reachability games via a reduction that introduces irrational transition probabilities in general.
We conclude that considering only MD strategies is sufficient to obtain the optimal value, i.e.\ $\sup$ and $\inf$ can be replaced with a $\max$ and $\min$ over MD strategies.
From this, we derive a system of inequalities that has a solution if and only if the optimal value satisfies $\ERisk^*\geq t$ for a given threshold $t$.
We conclude this section by observing that the satisfiability of this system of inequalities can be expressed as a sentence in the language of the reals with exponentiation.
In this way, we obtain the conditional decidability of the entropic risk threshold problem in SGs subject to Shanuel's conjecture.

Throughout this section, fix a game $\SG$, reward function $\reward$, state $\initialstate$, risk parameter $\gamma$, and risk basis $b$.
Omitted proofs can be found in \cref{app:basic_properties:proofs}.

\subsection{Determinacy and Optimality Equation} \label{sub:determinacy}
\begin{lemma}
	Stochastic games with $\Utility$ as payoff function are determined, i.e.
	\begin{equation*}
		{\inf}_{\stratmax \in \Strategies<\SG>} {\sup}_{\stratmin \in \Strategies<\SG>} \ExpectationSG<\SG, s><\stratmax, \stratmin>[b^{-\gamma \totalreward}]
	=
		{\sup}_{\stratmin \in \Strategies<\SG>}{\inf}_{\stratmax \in \Strategies<\SG>} \ExpectationSG<\SG, s><\stratmax, \stratmin>[b^{-\gamma \totalreward}].
	\end{equation*}
\end{lemma}
\begin{proof}
	This follows from the classical result on determinacy of Borel games \cite{martin1975borel}, see \cite{maitra2003stochastic} for a concrete formulation for stochastic games.
	In particular, the game  is zero-sum and $\Utility$ is a bounded, Borel-measurable function.
\end{proof}
As $\ERisk$ is obtained from $\Utility$ via a monotone function, also games with $\ERisk$ as payoff function are determined.
While $\ERisk^*$ is difficult to tackle directly due to its non-linearity, we can derive the following optimality equation for $\Utility^*$:
\begin{restatable}{lemma}{fixedpoint} \label{stm:fixed_point_equation}
	The optimal utility $\Utility^*$ is a solution of the following system of constraints:%
	\begin{equation} \label{eq:utility_fixed_point}
		v(s) = b^{-\gamma \reward(s)} \cdot \overline{\opt}^s_{a \in \stateactions(s)} \cdot {\sum}_{s' \in \States} \sgtransitions(s, a, s') \cdot v(s'),
	\end{equation}
	where $\overline{\opt}^s$ is $\min$ for a Maximizer state $s$ and $\max$ for a Minimizer state.
\end{restatable}
Unfortunately, $\Utility$ is not the unique or, at least, the pointwise smallest or largest fixed point of this equation system.
Consider the case where $\reward \equiv 0$, i.e.\ $b^{- \gamma \reward(s)} = 1$.
Here, every constant vector is a fixed point, however $\Utility^* \equiv 1$.
More generally, as the equations are purely multiplicative, for any fixed point $v$, every multiple $\lambda \cdot v$ is a fixed point, too.
\begin{example} \label{ex:example_multiple_fixpoints}
	Again consider the example of \cref{fig:running_example} with $b = 2$ and $\gamma = 1$.
	The (simplified) equations we get are:
	\begin{gather*}
		v_1 = 2^{-2} \cdot \min \{\tfrac{1}{2} v_2 + \tfrac{1}{2} v_3, v_4\} \qquad
		v_2 = 2^{-2} \cdot v_2 \qquad
		v_3 = v_3 \qquad
		v_4 = 2^{-4} \cdot v_3,
	\end{gather*}
	where $v_i$ corresponds to the value of $s_i$.
	First, for $v_2$, we observe that $v_2 = 0$ is the only valid assignment.
	Then, we have that $v_1 = 2^{-2} \cdot \min \{\tfrac{1}{2} 0 + \tfrac{1}{2} v_3, 2^{-4} v_3\} = 2^{-3} \cdot \min \{v_3, 2^{-3} v_3\}$.
	Clearly, this system is underdetermined and we obtain a distinct solution for any value of $v_3$.
\end{example}
To solve these issues, we need to define \enquote{anchors} of the equation.
We observe the resemblance of classical fixed point equations for stochastic systems.
In particular, for $\reward \equiv 0$, \cref{eq:utility_fixed_point} is the same as for reachability, \cref{eq:reachability_equation}.

\subsection{Reduction to Reachability}
\label{sub:reachability}
We define $S_0 = \{s \mid {\max}_{\stratmax} {\min}_{\stratmin} \ProbabilitySG<\SG, s><\stratmax, \stratmin>[\totalreward >0 ] = 0\}$ and $S_\infty = \{s \mid {\max}_{\stratmax} {\min}_{\stratmin} \ProbabilitySG<\SG, s><\stratmax, \stratmin>[\totalreward = \infty] = 1\}$ the set of states in which Maximizer cannot obtain a total reward of more than $0$ with positive probability against an optimal opponent strategy or ensure infinite reward with probability 1, respectively.
We show later on that these sets are simple to compute and MD strategies are sufficient.
Since $\reward(s) \geq 0$, all states in $s \in S_0$ necessarily have $\reward(s) = 0$.
Observe that $S_0$ may be empty, but then  $\States = S_\infty$ and so $\Utility^* = 0$, $\ERisk^* = \infty$.
Through these sets, we can connect optimizing the utility to a reachability objective.
\begin{restatable}{lemma}{reductionreachability} \label{stm:utility_reachability_equivalence}
	For any state $s$ in the game $\SG$, the optimal utility $\Utility^*$ is equal to the minimal probability of reaching the set $S_0$ from $s$ in game $\SG_R$, defined as follows:
	We add a designated sink state $\underline{s}$ (which may belong to either player and only has a self-loop back to itself) and define $\sgtransitions_R(s, a, s') = b^{-\gamma \reward(s')} \cdot \sgtransitions(s, a, s')$ for $s, s' \in \States$, $a \in \Actions(s)$ and $\sgtransitions_R(s, a, \underline{s}) = (1 - b^{- \gamma \reward(s)})$.
	There is a direct correspondence between optimal strategies.
\end{restatable}
We note that reachability games can also be reduced to our case: %
\begin{restatable}{lemma}{capturingreachability}\label{lem:capturing_reachability}
	For any game $\SG$ and (absorbing) reachability goal $T$, we have $\val_{\SG, \reach T}(s) = 1 -  \Utility^*_{\SG}(s)$ with reward $\reward(s) = \indicator{T}(s)$ and $\gamma = 1$.
\end{restatable}
We highlight that this reduction from entropic risk games to reachability games is \emph{not} an effective reduction in the computational sense, since $\SG_R$ comprises \emph{irrational} transition probabilities even for entirely rational inputs.
We discuss how to tackle this in the next section and first proceed to derive some useful properties from this correspondence.
\begin{restatable}{lemma}{utilitysmallestfixedpoint} \label{stm:utility_smallest_fixpoint}
	The optimal utility $\Utility^*$ is the pointwise smallest solution of 
	\begin{equation} \label{eq:utility_equation}
		\begin{gathered}
			v(s) = 0 \quad \text{for $s \in S_\infty$,} \qquad v(s) = 1 \quad \text{for $s \in S_0$, and} \\
			v(s) = \overline{\opt}^s_{a \in \stateactions(s)} b^{-\gamma \reward(s)} \cdot \ExpectedSumSG{\sgtransitions}{s}{a}{v} \quad \text{otherwise}
		\end{gathered}
	\end{equation}
\end{restatable}
Yet, there might be multiple fixed points to the system of equations. %
This is to be expected, since already reachability on MDPs exhibits this problem \cite{DBLP:journals/tcs/HaddadM18}.
We provide a discussion of these issues together with a sufficient condition for uniqueness in \cref{app:basic_properties:uniqueness}.

\subsection{Strategy Complexity} \label{sub:strategies}
By \cref{stm:utility_reachability_equivalence}, the optimal negative exponential utility is achieved by reachability-optimal strategies in $\SG_R$. 
With the known results on reachability \cite{condon1990algorithms},
this yields:
\begin{theorem}
	MD strategies are sufficient to optimize the negative exponential utility and thus also entropic risk.
	More precisely, for all SGs $\SG$, there is an MD strategy $\stratmax$ for the Maximizer such that
	$
		\ERisk_{\SG, s}^{\gamma *}={\inf}_{\stratmin \in \Strategies<\SG>} \ERiskSG<\SG, s><\gamma>[(\stratmax, \stratmin)]
	$
	and analogously for the Minimizer.
\end{theorem}
\begin{remark}
	We highlight that this means that this notion of risk is history independent:
	Which actions are  optimal does not depend on what has already \enquote{gone wrong}, but purely on the potential future consequences.
	This is in stark contrast to, e.g., conditional value-at-risk optimal strategies for total reward, which require exponential memory and switch to a purely expectation maximizing (i.e.\ risk-ignorant) behaviour after \enquote{enough} went wrong \cite{DBLP:conf/aaai/Meggendorfer22}.
\end{remark}
\subsection{System of Inequalities}
\label{sub:inequalities}
The problem we want to solve is deciding whether the Maximizer can ensure an entropic risk of at least $t$.
Unfortunately, the reachability game $\SG_R$ is not directly computable, since even for rational rewards $b^{- \gamma \reward(s)}$ may be irrational.
As such, we cannot use this transformation directly to prove decidability or complexity results and need to take a different route.
Analogous to the classical solution to reachability, we first convert the problem to a system of inequalities.
Intuitively, we replace every $\max$ with $\geq$ for all options and dually $\min$ with $\leq$ (again, recalling that Maximizer wants to minimize the value in $\SG_R$).
Formally, we consider the following:
\begin{equation} \label{eq:utility_inequalities}
	\begin{gathered}
		v(\initialstate) \leq b^{-\gamma t}, \qquad v(s) = 0 \quad \text{for $s \in S_\infty$,} \qquad v(s) = 1 \quad \text{for $s \in S_0$,} \\
		v(s) \leq b^{-\gamma \reward(s)} \cdot \ExpectedSumSG{\sgtransitions}{s}{a}{v} \quad \text{for $s \in \StatesMax$, $a \in \Actions(s)$,} \\
		v(s) \geq b^{-\gamma \reward(s)} \cdot \ExpectedSumSG{\sgtransitions}{s}{a}{v} \quad \text{for $s \in \StatesMin$, $a \in \Actions(s)$, and} \\
		\lOr_{a \in \stateactions(s)} v(s) = b^{-\gamma \reward(s)} \cdot \ExpectedSumSG{\sgtransitions}{s}{a}{v} \quad \text{for $s \in \States$}
	\end{gathered}
\end{equation}
Observe that this essentially is the decision variant to the standard quadratic program for reachability applied to $\SG_R$ \cite{condon1992complexity}.
\begin{restatable}{lemma}{utilityinequalities}
	The system of equations \ref{eq:utility_inequalities} has a solution if and only if $\ERisk^* \geq t$.
\end{restatable}

\subsection{Decidability Subject to Shanuel's Conjecture}
From \cref{eq:utility_inequalities}, we  obtain a  conditional decidability result for the general case:
\begin{theorem} \label{thm:schanuel}
	Let all quantities, i.e.\ rewards, transition probabilities, the risk-aversion factor $\gamma$, and the basis $b$ be given as formulas in the language of reals with exponentiation (i.e.\ with functions $+$, $\cdot$, and $\exp \colon x \mapsto e^x$).
	Then, the entropic risk threshold problem for SGs is decidable subject to Schanuel's conjecture.
\end{theorem}
\begin{proof}
	In this case, the existence of a solution to \cref{eq:utility_inequalities} can also be expressed as a sentence in the language of the reals with exponentiation.
	The corresponding theory is known to be decidable subject to Shanuel's conjecture (see e.g.\ \cite{lang1966introduction}) as shown by \cite{MacintyreWilkie1996}, and decidability of this theory is equivalent to the so-called \enquote{weak Schanuel's conjecture}.
\end{proof}
In particular, this allows us to treat instances with basis $b=e$.
Yet, even if all rewards, transition probabilities, and $\gamma$ are given as rational values, but the basis $b$ equals $e$, we do not know how to check the satisfiability of \cref{eq:utility_inequalities} without relying on the theory of the reals with exponentiation.
Note, however, that we do not need the \enquote{full power} of the exponential function:
All values appearing in an exponent in \cref{eq:utility_inequalities} are constants.
So, the restricted exponential function that agrees with $\exp$ on a closed interval $[a_1,a_2]$ and is zero outside of this interval is sufficient.
The theory of the reals with restricted exponentiation has some additional nice properties compared to the theory of the reals with full exponentiation:
For example, it allows for quantifier elimination by \cite{van1994elementary} and related works.
Nevertheless, this does not allow us to immediately obtain an unconditional decidability result.

\section{The Algebraic Case}
\label{sec:algebraic}
If  all occurring values are rational, then all numbers of the system of inequalities \cref{eq:utility_inequalities} are algebraic.
The results of this section establish that the threshold problem for such instances is decidable.
A detailed exposition of the results can be found in Appendix \ref{app:algebraic};
 an overview of the complexity results can also be found in \cref{tbl:overview}. 
Formally, we define:
\begin{definition}
	An \emph{algebraic instance} of the entropic risk threshold problem is an instance where all occurring values, i.e.\ the transition probabilities of the game $\SG$, all rewards assigned by the reward function $r$, the risk-aversion parameter $\gamma$, the basis $b$, and the threshold $t$, are rational and encoded as the fraction of co-prime integers in binary.
\end{definition}

In general, for algebraic instances, there is a  reduction of our problem to the existential theory of the reals, leading to the following result where $\exists \Reals$ denotes the complexity class of problems that are polynomial-time reducible to the existential theory of the reals:
\begin{restatable}{theorem}{generalalgebraic}\label{thm:generalalgebraic}
	For algebraic instances, the entropic risk threshold problem is decidable in $\exists \Reals$ and thus in PSPACE.
\end{restatable}

\begin{remark}
	We note that already in simple MCs, solutions might not be rational.
	Consider a MC yielding the a distribution over total rewards of $\{0 \mapsto \frac{1}{2}, 1 \mapsto \frac{1}{2}\}$.
	For $b = 2$ and $\gamma = 1$, $\ERisk^*$ equals $- \log_2(\tfrac{1}{2} \cdot 2^0 + \tfrac{1}{2} \cdot 2^{-1}) = \log_2(\tfrac{4}{3}) = 2 - \log_2(3)$, an irrational number.
	Note that the input comprises only small, rational values.
	The threshold problem for the negative exponential utility suffers from the same issue:
	For $b = 2$ and $\gamma = \tfrac{1}{2}$, $\Utility^*$ on the same MC equals $\tfrac{1}{2} \cdot 2^0 + \tfrac{1}{2} \cdot 2^{-\frac{1}{2}} = \tfrac{1}{2(1 + \sqrt{2})}$, again an irrational number.

On the one hand, this indicates that already for MCs, there is no obvious improvement of the complexity upper bound.
Although  the negative exponential utility is the solution to a linear system of equations in the MC case, the main obstacle for efficient computations is the complicated form of the occurring numbers.
On the other hand, the \enquote{bottleneck} for complexity lies already in the MC case:
If $C$ is the threshold problem  for  MCs, the problem for SGs can be solved in $(\Sigma_2^P)^C \cap (\Pi_2^P)^C$ due to determinacy and optimality of MD strategies:
In a game $\SG$, we have $\ERisk^{\ast}_{\SG}\geq t$ if, for all MD-strategies for the Minimizer, there is an MD-strategy for  the Maximizer such that in the resulting MC $\ERisk\geq t$, or equivalently, if there is an MD-strategy for the Maximizer such that for all MD-strategies for the Minimizer  $\ERisk\geq t$  in the resulting MC.
This puts the problem into $(\Sigma_2^P)^C$ and  $(\Pi_2^P)^C$.
\end{remark}
For \cref{thm:generalalgebraic}, we use the standard decision procedure for the existential theory of the reals as a \enquote{black box} and do not make use of the special form of our problem.
To exploit the specific structure of the system of inequalities, we note that for explicit computations on algebraic numbers the following two  quantities are relevant for the resulting computational complexity:
Firstly, the degree of the field extension of $\mathbb{Q}$ in which the computation can be carried out. Secondly, the bitsize of the coefficients of the minimal polynomials of the involved algebraic numbers (see, e.g., \cite{AdlerB94,Beling01}).  Alternatively, the bitsize of the representations of the algebraic numbers in a fixed basis of the  field extension in which the computations can be carried out can be used. Note that the size of the basis is precisely the degree of that field extension.
Motivated by these observations, we  consider \emph{small algebraic instances}, which allow us to prove that all occurring algebraic numbers have a sufficiently small representation.%
\begin{definition}
	A \emph{small algebraic instance} of the entropic risk threshold problem consists of a SG $\SG$ with rational transition probabilities, an integer reward function $r$, a rational risk-aversion parameter $\gamma$, a rational basis $b$, and a rational threshold $t$.
	Moreover, the rewards, $\gamma$, and $t$ are encoded in unary, and as the fraction of co-prime integers encoded in unary, respectively.
	The remaining rational numbers are encoded as the fraction of co-prime integers in binary.
	If $\SG$ is an MDP or a MC, we call the instance a small algebraic instance of an MDP or a MC.
\end{definition}
\begin{remark}%
	For simplicity, we assume for small algebraic instances that all rewards are in $\Naturals$.
	If this is not the case, we can multiply all rewards with the least common multiple $D$ of the denominators of the rewards and  use a new risk-aversion parameter $\gamma' = \gamma / D$.
	The resulting negative exponential utility is not affected by this transformation.
	The change of the optimal entropic risk  by a factor of $D$ can  be addressed  by also rescaling the threshold $t' = t \cdot D$.
	Nevertheless, note that this affects the encoding size of the risk-aversion factor $\gamma$.
\end{remark}
Relying on algorithms for explicit computations in algebraic numbers \cite{AdlerB94,Beling01}, we obtain:
\begin{restatable}{theorem}{smallalgebraicthreshold}\label{thm:smallalgebraicthreshold}
	For small algebraic instances, the entropic risk threshold problem: (a) belongs  to $\mathsf{NP}\cap\mathsf{coNP}$ for SGs; and (b) can be solved in polynomial time for MDPs or MCs.
\end{restatable}

While the mentioned results concern the threshold problem, we can even go a step further in small algebraic instances of MCs.
Here, the system of inequalities simplifies to a linear system of equations, which we can solve \emph{explicitly} in the algebraic numbers.
For small algebraic instances, this is possible in polynomial time yielding the following result.
\begin{restatable}{theorem}{smallalgebraiccomputation}\label{thm:algebraic_optimal_value}
	For small algebraic instances,  an explicit representation of $\Utility^*$ can be computed in: (a) polynomial time for MCs; and (b)
	 in polynomial space for  SGs and MDPs.
\end{restatable}

\section{Approximation Algorithms} \label{sec:approximation}

The results of the previous section suggest, depending on the form of the input, a polynomial-space algorithm or even worse in the general case.
Clearly, this is somewhat unsatisfactory for practical applications.
Recall that the difficulties are due to the occurring irrational transition probabilities.
In the hope that we can work with approximations of these numbers, we now aim to identify an approach which allows us to approximate the correct answer, i.e.\ compute a value close to the optimal entropic risk that the Maximizer can ensure.
Again, fix an SG $\SG$, reward function $\reward$, risk parameter $\gamma$, and risk basis $b$ throughout this section.
Then, given a precision $\varepsilon > 0$, we aim to compute a value $v$ such that $\abs{{\ERisk^*} - v} < \varepsilon$, i.e.\ an approximation with small absolute error.

Since entropic risk is the logarithm of utility, we need to obtain an approximation of $\Utility^*$ to a sufficiently small \emph{relative} error.
Concretely, we need to compute a value $v_U$ such that $b^{- \gamma \varepsilon} \leq v_U / \Utility^* \leq b^{\gamma \varepsilon}$.
Then, $v = - \frac{1}{\gamma} \log_b(v_U)$ yields an approximation, since
\begin{gather*}
	{\ERisk^*} - v = - \tfrac{1}{\gamma} \log_b(\Utility^*) + \tfrac{1}{\gamma} \log_b(v_U) = \tfrac{1}{\gamma} \log_b(v_U / \Utility^*) = (*) \\
	(*) \geq \tfrac{1}{\gamma} \log_b(b^{-\gamma \varepsilon}) = -\varepsilon \quad \text{ and } \quad (*) \leq \tfrac{1}{\gamma} \log_b(b^{\gamma \varepsilon}) = \varepsilon.
\end{gather*}
(When we are interested in a concrete value for $v$, we need to determine $v_U$ with a slightly higher precision and then approximate $\log_b(v_U)$ sufficiently.)
Now, in order to approximate $\Utility^*$, we still need to deal with a system comprising potentially irrational transition probabilities.
We argue that the occurring values $b^{-\gamma \reward(s)}$ can be \enquote{rounded} to a sufficient precision while keeping the overall relative error small.
Using techniques from \cite{chatterjee2012robustness}, in \cref{app:approximation} we provide an effective way to compute a game $\SG_{\approx}$, which behaves \enquote{similarly} to the reachability game $\SG_R$ from \cref{stm:utility_reachability_equivalence}.
Once $\SG_{\approx}$ is computed, we can employ classical solution methods, such as linear equation solving for MCs, linear programming for MDPs, or, e.g., quadratic programming for SGs leading to an algorithm in polynomial space for SGs.
\begin{theorem}\label{thm:approximation}
	In MCs and MDPs, the optimal value $\ERisk^\ast$ can be approximated up to an absolute error of $\varepsilon$ in time polynomial in the size of the system, $-\log(\varepsilon)$, $\log  b$, $\gamma\cdot r_{\max}$, and $1/(\gamma\cdot r_{\min})$, where $r_{\max}$ and $r_{\min}$ are the largest and smallest occurring non-zero rewards, respectively.
	For SGs, this is possible in polynomial space.
\end{theorem}
In particular, for fixed $b$ and $\gamma$, and bounded rewards (both from above and below), we obtain a PTIME solution for MC and MDP.
In general, the procedure is exponential for SG.
Alternatively, we can also apply different approaches such as value iteration \cite{SG_EC_LICS2023}.
\begin{remark}
	We note the connection to the small algebraic case:
	The \enquote{limiting factor} in both cases is the (size of the) product of $\gamma$ and the state rewards.
	If these are fixed or given in unary, respectively, the complexity of our proposed algorithms is significantly reduced.
\end{remark}
As a final note, recall that we do not assume $\gamma$ or the transition probabilities to be rational.
We only require that we can expand their binary representation to arbitrary precision.
Then, we can conservatively approximate their logarithm to evaluate the required rounding precision and approximate the transition probabilities of $\SG_{\approx}$ in the same way.

\section{Conclusion}

We applied  the entropic risk  to total rewards in SGs to capture risk-averse behavior in these games. The objective forces agents to achieve a good overall performance while keeping the chance of particularly bad outcomes small.
We showed that SGs with the entropic risk as payoff function are determined and admit optimal MD-strategies. This reflects
the
time-consistency of  entropic risk and makes   entropic risk   an appealing objective as, in contrast, the optimization of other risk-averse objective functions that have been studied on MDPs in the literature require strategies with large memory or complicated randomization.

Computationally,
 difficulties arise due to the involved exponentiation leading to irrational or even transcendental numbers. For the general case, we obtained decidability of the threshold problem  only subject to Shanuel's conjecture while for purely rational inputs, the problem can be solved via a reduction to the existential theory of the reals. 
 Additional restrictions on the encoding of the input allowed us to obtain better upper bounds. Further, we provided an approximation algorithm for the optimal value.
For an overview of the results, see \cref{tbl:overview}.

A question  that is left open is whether the entropic risk threshold problem for algebraic instances of MCs can be solved more efficiently than by the polynomial-time reduction to the existential theory of the reals. This case constitutes a bottleneck in the complexity.
Furthermore, we worked with non-negative rewards, which made a
reduction from games with the entropic risk objective to reachability games possible. Dropping the restriction to non-negative rewards constitutes an interesting direction of future research, in which additional difficulties arise and a reduction to reachability is not possible anymore. 
A further direction for future work is the experimental evaluation of the proposed algorithms to assess their practical applicability as well as to investigate the behavior of the resulting optimal strategies.
In particular, it might be interesting to investigate the \enquote{cost} of risk-awareness, namely how much the expected total reward of a risk-aware strategy differs from a purely expectation maximizing one on realistic systems.

\bibliographystyle{IEEEtran}
\bibliography{main}

\clearpage
\appendix
\section{Additional Discussion for \ref{sec:basic_porperties}} \label{app:basic_properties}

\subsection{Omitted Proofs} \label{app:basic_properties:proofs}

\fixedpoint*

\begin{proof}
	Fix an arbitrary Maximizer state $s$ (the statement follows analogous for Minimizer states).
	Let $\strategy = (\stratmax^*, \stratmin^*)$ an optimal strategy profile, ensuring a value of $\Utility^*$ in $s$.
	Observe that
	\begin{equation*}
		\Utility^*(s) = \ExpectationSG<\SG, s><\strategy>[b^{-\gamma \totalreward}] = \integral<\infinitepath \in \Infinitepaths<\SG>>{b^{- \gamma \sum_{i=1}^{\infty} \reward(\infinitepath_i)}}{\ProbabilitySG<\SG, s><\strategy>}.
	\end{equation*}
	We can rewrite $b^{- \gamma \sum_{i=1}^{\infty} \reward(\infinitepath_i)} = b^{-\gamma \reward(\infinitepath_1)} \cdot b^{-\gamma \sum_{i=2}^\infty \infinitepath_i}$.
	Inserting this equality in the above equation, splitting the set of paths by the action taken and the successor state, and shifting indices by 1 in the integral, we get
	\begin{equation*}
		\ExpectationSG<\SG, s><\strategy>[b^{-\gamma \totalreward}] = b^{-\gamma \reward(s)} \sum_{a \in \stateactions(s), s' \in \States} \sgtransitions(s, a, s') \ExpectationSG<\SG, s'><\strategy_{(s, a)}>[b^{-\gamma \totalreward}],
	\end{equation*}
	where $\strategy_{(s, a)}$ is the strategy profile both players follow after seeing $(s, a)$.
	Now, observe that $\ExpectationSG<\SG, s'><\strategy_{(s, a)}>[b^{-\gamma \totalreward}] = \Utility^*(s')$ for any $s'$ where $\sgtransitions(s, a, s') > 0$:
	If it were  different, the respective player could employ a better strategy in this particular case, ensuring a different value in $s$.
	Thus,
	\begin{equation*}
		\Utility^*(s) = b^{-\gamma \reward(s)} \sum_{a \in \stateactions(s), s' \in \States} \sgtransitions(s, a, s') \Utility^*(s').
	\end{equation*}
	Now, it immediately follows that taking any optimal action ensures the best outcome, proving the result.
\end{proof}

\reductionreachability*

\begin{proof}[Sketch]
	We first provide a brief proof sketch.
	Essentially, we first show that $\States_0$ and $\States_\infty$ are \enquote{sinks} of the play, i.e.\ that any optimal strategy can and will keep the play inside either set once it is reached.
	Moreover, they are the only sinks in that sense, i.e.\ under any optimal strategy either of them is reached with probability one.

	The remainder of the proof then shows that the \enquote{discounted} reachability (achieved by the transitions to $\underline{s}$) corresponds to the negative utility and that optimal strategies correspond.
	From $\SG$ to $\SG_R$, we split all paths that reach $\States_0$ by their prefix, observing that the total reward obtained by the path equals that of the prefix (since once in $\States_0$ no more reward is obtained), and all other paths obtain a reward of $\infty$ with probability 1, meaning they reach $\underline{s}$ with probability 1, too.
	From $\SG_R$ to $\SG$, we again split the paths reaching $\States_0$ by their prefix and notice that the probability of those prefixes exactly corresponds to the discounted total reward they would achieve in $\SG$.
\end{proof}

\begin{proof}
	Before we tackle the actual statement, we discuss some important properties of $S_0$.

	\underline{Properties of $S_0$ and $S_\infty$:}
	We first show that $S_0$ represents a \enquote{sink} of the play, in the sense that any optimal Minimizer strategy will keep the play inside $S_0$.
	Moreover, this also proves that such a strategy always exists.
	Observe that by definition there exists a (memoryless deterministic) strategy of Minimizer ensuring that Maximizer can only obtain a total reward of $0$ for any state in $S_0$ and consequently a utility of $1$.
	This already shows that by keeping the play inside $S_0$, Minimizer acts optimally (there is no better outcome).
	One can also show that remaining inside $S_0$ is the only optimal play, however we do not require this statement.
	We analogously show a similar property for $S_\infty$, i.e.\ an optimal strategy of the Maximizer can always keep the play in $S_\infty$.
	We note that these strategies are rather simple to obtain, indeed for any state $s \in S_0$ we can just take any action $a \in \stateactions(s)$ with $\support(\sgtransitions(s, a)) \subseteq S_0$, dually for $S_\infty$.
	In particular, they are memoryless and deterministic.

	Finally, we show that $S_0$ and $S_\infty$ are the only two sinks, i.e.\ $\ProbabilitySG<\SG, \initialstate><\strategy^*>[\reach (S_0 \union S_\infty)] = 1$ for any utility optimal strategy profile $\strategy^*$.
	(Recall that both sets are absorbing for optimal strategies.)
	Suppose there exists some set of states $S'$ disjoint from both $S_0$ and $S_\infty$ in which paths remain forever with positive probability under $\strategy^*$.
	There necessarily exists a subset of $S'$ which are visited infinitely often with probability 1 under $\strategy$, let $S'$ equal this subset.
	If all states in $S'$ have a reward of $0$, then necessarily $S' \subseteq S_0$:
	Since $\strategy^*$ is optimal but nevertheless a reward of $0$ is obtained in $S'$, the Minimizer can ensure a total reward of $0$ in $S'$.
	Dually, if there is some state with non-zero reward in $S'$, this state occurs infinitely often with probability 1, ensuring an infinite expected reward, i.e.\ $S' \subseteq S_\infty$.

	To summarize, we have that $S_0$ and $S_\infty$ are the unique \enquote{sinks} of the game under \emph{optimal strategies}.
	With this in place, we can proceed with the main proof.

	\medspace

	We need to show equivalence of optimal values and a correspondence between strategies.
	In particular, we shall prove that any optimal strategy profile in $\SG$ is directly equivalent to an optimal strategy in $\SG_R$, and for the other direction we only need to modify the decisions of Minimizer on $S_0$ to keep the play inside.
	We note a subtlety: We first prove that any optimal strategy in $\SG$ achieves a reachability probability of $\Utility^*$ in $\SG_R$.
	This alone does not yet show that $\Utility^*$ is the optimal reachability probability.
	We then also prove that optimal strategies for reaching $S_0$ on $\SG_R$ corresponds to a strategy achieving $\Utility^*$ on $\SG$.
	In the following, we slightly abuse notation and identify strategies on $\SG$ with strategies on $\SG_R$, extended by the meaningless choice in $\underline{s}$.

	\underline{From $\SG$ to $\SG_R$:}
	We start with the forward direction.
	Let $\strategy^* = (\stratmax^*, \stratmin^*)$ be utility-optimal strategies in $\SG$.
	We shall argue that $\Utility^* = \Utility(\strategy^*) = \ProbabilitySG<\SG_R, \initialstate><\strategy^*>[\reach S_0]$.
	To this end, we investigate the set $\reach S_0$ more closely.
	Every path which reaches a set does so after finitely many steps.
	Thus, we split this set by the finite prefix after which paths reach their respective destination (a countable set).
	Formally, for a prefix $\finitepath$, we write $\reach_\finitepath S_0$ to denote the set of all paths which begin with $\finitepath$ and reach $S_0$ exactly at the end of $\finitepath$.
	For simplicity, we define this set to be empty if $\finitepath$ does not end in $S_0$ or already reached $S_0$ earlier.
	This way, we have in general that $\reach S_0 = \Union_{\finitepath} \reach_\finitepath S_0$ and the sets are pairwise disjoint.
	The probability of reaching $S_0$ in $\SG_R$ thus can be written as $\sum_{\finitepath} \ProbabilitySG<\SG_R, \initialstate><\strategy^*>[\reach_\finitepath S_0]$, the probability of all finite paths reaching $S_0$ (observe the similarity to the general proof of measurability for $\reach T$).
	By inserting the definition of transition probabilities in $\SG_R$ and reordering, we obtain $\prod_{i=1}^{\cardinality{\finitepath}} b^{-\gamma \reward(\finitepath_i)} \cdot \ProbabilitySG<\SG, \initialstate><\strategy^*>[\reach_\finitepath S_0]$.
	Note that $\underline{s}$ is absorbing and thus cannot occur on any path in $\reach_\finitepath S_0$.

	Recall that no rewards are obtained once a path reaches $S_0$ (since $\stratmin^*$ is optimal), thus the total reward of any path $\infinitepath \in \reach_\finitepath S_0$ exactly equals $\totalreward(\infinitepath) = \sum_{i=1}^{\cardinality{\finitepath}} \reward(\finitepath_i)$.
	Since $\Utility$ is defined as expectation, we can similarly split up the set of all runs in $\SG$ in a linear manner first into $\reach S_0$ and $\overline{\reach S_0}$ and then further split up $\reach S_0$ as above.
	Together, we obtain
	\begin{multline*}
		\Utility^* = \ExpectationSG<\SG, \initialstate><\strategy^*>[b^{-\gamma \totalreward}] = \\
			{\sum}_{\finitepath} b^{- \gamma \totalreward(\finitepath)} \cdot \ProbabilitySG<\SG, \initialstate><\strategy^*>[\reach_\finitepath S_0] + \integral<\infinitepath \in \overline{\reach S_0}>{b^{-\gamma \totalreward(\infinitepath)}}{\ProbabilitySG<\SG, \initialstate><\strategy^*>}.
	\end{multline*}
	Observe that the left hand side exactly equals the probability of $\reach S_0$ in $\SG_R$ as we argued above.
	It thus remains to show that the remaining integral has a value of zero.
	Here, we need to exploit the optimality of $\stratmax^*$.
	Recall that in this case almost all paths which do not reach $S_0$, i.e.\ $\overline{\reach S_0}$, instead reach $S_\infty$.
	Thus, for all these paths the Maximizer ensures an infinite total reward, corresponding to a utility of $0$.

	In summary, every optimal strategy profile $\strategy^*$ for $\SG$ reaches $S_0$ in $\SG_R$ with probability $\Utility^*$.

	\underline{From $\SG_R$ to $\SG$:}
	Now, let $\strategy^* = (\stratmax^*, \stratmin^*)$ be (memoryless deterministic) reachability-optimal strategies in $\SG_R$.
	Note that transition probabilities on $S_0$ are the same for $\SG$ and $\SG_R$, since $\reward(s) = 0$ on all these states.
	Thus, there exists a Minimizer strategy which keeps the play inside $S_0$ in $\SG_R$ and we assume w.l.o.g.\ that $\stratmin^*$ behaves in this way.
	This clearly does not influence the probability of reaching $S_0$ in the first place and thus $\stratmin^*$ remains an optimal strategy.
	Moreover, for $S_\infty$ observe that there exists a Maximizer strategy in $\SG$ ensuring that an infinite reward is obtained with probability 1.
	Then, following this strategy on $S_\infty$ in $\SG_R$ ensures that $\underline{s}$ is reached almost surely.
	This means that once $S_\infty$ is reached, Maximizer can ensure that $S_0$ is never reached, since $\underline{s}$ is absorbing.

	We divide the set of all possible infinite paths in $\SG_R$ into three groups, namely (i)~those which reach $S_0$, (ii)~those which reach $\underline{s}$, and (iii)~all others (which, as we will show, turn out to be a null set).
	Note that for (ii) considering those which reach $S_\infty$ would be wrong, since $\underline{s}$ can also be reached from other states (to be precise, from any state with $\reward(s) > 0$).
	For the first kind, we apply the same reasoning as above, obtaining that $\ProbabilitySG<\SG_R, \initialstate><\strategy^*>[\reach_\finitepath S_0] = \prod_{i=1}^{\cardinality{\finitepath}} b^{-\gamma \reward(\finitepath_i)} \cdot \ProbabilitySG<\SG, \initialstate><\strategy^*>[\reach_\finitepath S_0]$.
	For the second case, observe there is no direct equivalent of $\reach \{\underline{s}\}$ in $\SG$.
	Instead, we show that the set of states for which the optimal probability to reach $S_0$ is $0$ exactly is $S_\infty \union \{\underline{s}\}$.
	As argued above, this certainly is true for all states in $S_\infty \union \{\underline{s}\}$.
	Thus, choose some state $s \in \States \setminus (S_0 \union S_\infty)$.
	For this state, there exists a Minimizer strategy in $\SG$ ensuring that a finite total reward is obtained with non-zero probability (otherwise, $s$ would belong to $S_\infty$).
	This can only be the case if $S_0$ is reached with non-zero probability under this strategy in $\SG$.
	By replicating this strategy on $\SG_R$, Minimizer can as well ensure a non-zero probability of reaching $S_0$, since every path in $\SG$ also is possible in $\SG_R$ (only with a potentially decreased probability).
	We now want to argue that the third kind has measure zero.
	Consider the Markov chain $\MC$ induced by $\strategy^*$ (which is finite since $\strategy^*$ is memoryless).
	If $\MC$ had any BSCC containing states in neither $S_0$ nor $S_\infty \union \{\underline{s}\}$, we could adapt the Minimizer strategy to follow the previous strategy in this BSCC, increasing the probability to reach $S_0$.
	Recall that Minimizer wants to maximize the probability of reaching $S_0$, so this contradicts the optimality of $\strategy^*$.
	Consequently, the induced Markov chain only has BSCCs which are subsets of either $S_0$ or $S_\infty \union \underline{s}$, meaning almost all paths reach either of these.
	Together we obtain the result, i.e.\ that $\strategy^*$ obtains a utility of $\Utility^* = \ProbabilitySG<\SG_R, \initialstate><\strategy^*>[\reach_\finitepath S_0]$ in $\SG$.

	\underline{Combining the results:}
	From the first part, we get that any utility optimal strategy obtains a reachability probability of $\Utility^*$ in $\SG_R$.
	From the second part, we dually get that the optimal reachability probability in $\SG_R$ equals $\Utility^*$ and the optimal strategies for the reachability correspond to strategies obtaining $\Utility^*$ in $\SG_R$.
	Consequently, memoryless deterministic strategies are sufficient to optimize utility and we can obtain both the optimal utility as well as optimal strategies by solving the reachability game $\SG_R$.
\end{proof}

\capturingreachability*

\begin{proof}
	We assume w.l.o.g.\ that all states of $T$ are absorbing in $\SG$ (can be achieved in linear time without changing the reachability probability).
	Now, observe that the transition structure of $\SG_R$ completely agrees with $\SG$ except on $T$, where every state has a self-loop and a transition to $\underline{s}$.
	As we argued in the proof before, under optimal strategies almost all runs either reach $S_0$ or $S_\infty \union \underline{s}$ in $\SG_R$.
	Thus, since Maximizer is minimizing the probability to reach $S_0$, the probability to reach $S_\infty \union \underline{s}$ is maximized.
	To conclude, observe that (i)~$T \subseteq S_\infty$, in particular $S_\infty$ additionally exactly contains all states from which the Maximizer can force the play into $T$ with probability one, and (ii)~$\underline{s}$ cannot be reached without reaching $S_\infty$ first.
	Hence, the optimal probability of reaching $S_\infty \union \underline{s}$ in $\SG_R$ equals the probability to reach $S_\infty$ which in turn is the optimal probability to reach $T$ in $\SG$.
\end{proof}

\utilitysmallestfixedpoint*

\begin{proof}
	Follows directly from \cref{stm:utility_reachability_equivalence} combined with standard result on reachability for stochastic games \cite{DBLP:conf/spin/ChatterjeeH08}.
	Note that maximizing and minimizing reachability in stochastic games is equivalent, since we can simply swap the players.
\end{proof}

\utilityinequalities*

\begin{proof}
	\underline{If:}
	In this case, i.e.\ $\ERisk^* \geq t$, we have $\Utility^* \leq b^{-\gamma t}$.
	Thus, $\Utility^*$, which is a solution to \cref{eq:utility_fixed_point}, immediately satisfies all equations.

	\underline{Only If:}
	Observe that the first four equations ensure that any solution actually solves \cref{eq:utility_equation}.
	Since $\Utility^*$ is the pointwise smallest solution by \cref{stm:utility_smallest_fixpoint}, having a vector $v$ which satisfies the last inequality ensures that $\Utility^*$ does so, too.
	We conclude by noticing the equivalence of the first four equations to optimal solutions of the quadratic program for reachability, see e.g.\ \cite{condon1992complexity} for further information.
%
\end{proof}

\subsection{Uniqueness of Fixpoint} \label{app:basic_properties:uniqueness}

We identify a condition on $\SG$ which ensures the reachability game $\SG_R$ being \enquote{stopping} \cite{shapley1953stochastic} (or \enquote{halting}), which implies that the fixed point is unique.
\begin{restatable}{lemma}{uniquefixedpoint} \label{stm:utility_unique_fixpoint_condition}
	\cref{eq:utility_equation} has a unique solution if for every strategy profile $\strategy$ and state $s$, we have $\ProbabilitySG<\SG, s><\strategy>[\reach (S_0 \union \{s \mid \reward(s) > 0\})] > 0$, i.e.\ $S_0$ or a state that yields some non-zero reward is always reached with some positive probability.
\end{restatable}
\begin{proof}
	Under the assumption, the reachability game $\SG_R$ is \emph{stopping}, i.e.\ no matter what either player does, the game eventually stops with probability 1.
	Uniqueness of the equations for reachability on $\SG_R$ follows from standard results on reachability \cite{shapley1953stochastic}.

	Intuitively, given the assumptions, the reachability game $\SG_R$ will reach $S_0 \union \{\underline{s}\}$ with positive probability from any state (since there is a positive probability to directly transition to $\underline{s}$ in $\SG_R$ for any state with $\reward(s) > 0$).
	Now, observe that any set can only be reached with positive probability if there exists a path of length at most $\cardinality{S}$ to it.
	This implies that every $\cardinality{S}$ steps, there is some probability of reaching those states.
	If we repeat this ad infinitum, these states will be reached with probability 1.
	One can prove that the fixed point iteration corresponding to \cref{eq:utility_equation} applied $\cardinality{S}$ times is a contraction and thus has a unique fixed point.
\end{proof}
However, we hardly can expect to find a necessary condition:
Since our problem is essentially equivalent to reachability, the remarks of \cite{DBLP:journals/iandc/EisentrautKKW22} transfer to our case.
In particular, they give an if and only if condition for uniqueness of fixed points together with reasoning why we cannot easily decide this condition through graph analysis.
On the positive side, we can efficiently compute both $S_0$ and $S_\infty$ as well as check the stopping criterion directly on $\SG$.
\begin{restatable}{lemma}{quadtratictime}
	The sets $S_0$ and $S_\infty$ as well as the property of \cref{stm:utility_unique_fixpoint_condition} can be obtained / checked in time quadratic w.r.t.\ the number of states and transitions (assuming that the set $\{s \mid \reward(s) > 0\}$ can be determined in at most quadratic time).
\end{restatable}
\begin{proof}
	Decidability of such qualitative properties is well-known (e.g.\ through attractor computations on the game graph).
	For completeness, we rephrase the problems in terms of \cite{DBLP:conf/lics/AlfaroH00}.
	Observe that $S_0$ is equivalent to the condition \enquote{at every step we have $\reward(s) = 0$ with probability 1} and $S_\infty$ means \enquote{we infinitely often see $\reward(s) > 0$ with probability 1}, i.e.\ a \emph{sure safety} and \emph{sure Buchi} condition, respectively.
	Similarly, for the condition of \cref{stm:utility_unique_fixpoint_condition}, consider its negation:
	We ask if there exists a strategy under which the probability to reach $S_0 \union \{s \mid \reward(s) > 0\}$ is zero -- a sure safety query.
	All three queries can be easily answered using the approaches of \cite{DBLP:conf/lics/AlfaroH00}.

	We note that all computations only require knowledge of the transition structure of the underlying (hyper-)graph and the set of states with $\reward(s) > 0$ (which we assumed to be computable in quadratic time).
\end{proof}

\section{The Algebraic Case}
\label{app:algebraic}

\subsection{Preliminaries: Vector Spaces and Field Extensions}
We assume some familiarity with the algebraic field extensions. We introduce the main concepts briefly and provide our notation.
For more details, see, e.g., \cite{DBLP:books/daglib/0070572}.

\noindent
\textbf{Vectors.}
For a vector $v\in \mathbb{K}^n$ of a field $\mathbb{K}$, we denote its components by $v_0,\dots,v_{n-1}$.
Whenever comparing two vectors  $x$ and $y$ over an ordered field  by $x \leq y$ or taking their $\max$ or $\min$, this is understood point-wise.
We often encounter functions assigning values to (finitely many) states;
for convenience, we assume an implicit (arbitrary but fixed) numbering of each set of states and identify such functions with the corresponding vectors.

\noindent
\textbf{Field extensions.}
Given a complex number $\alpha$, the smallest field extension $\mathbb{F}$ of $\mathbb{Q}$ containing $\alpha$, i.e.\ with $\alpha\in \mathbb{F}$, is denoted by $\mathbb{Q}(\alpha)$.
The number $\alpha$ is called algebraic if there is a non-constant rational polynomial $P\in \mathbb{Q}[X]$ with $P(\alpha)=0$.
A complex number that is not algebraic is called transcendental.
The degree of an algebraic number $\alpha$ is the smallest degree of a non-zero polynomial $P$ with $P(\alpha)=0$.
There is a unique non-zero polynomial $P$ with $P(\alpha)=0$ of minimal degree with leading coefficient $1$, which is called the minimal polynomial of $\alpha$.
For an algebraic number $\alpha$ of degree $q$, $\mathbb{Q}(\alpha)$ is a $\mathbb{Q}$-vector space of dimension $q$.
One basis is given by $(\alpha^0,\dots, \alpha^{q-1})$.
For any basis $B=(b_0,\dots,b_{q-1})$ of $\mathbb{Q}(\alpha)$ as $\mathbb{Q}$-vector space
and any number $\beta\in \mathbb{Q}(\alpha)$, there is a unique vector $v\in \mathbb{Q}^q$ such that $\beta=\sum_{i=0}^{q-1}v_i\cdot b_i$.
We call $v$ the representation of $\beta$ in basis $B$.
Vice versa, given a vector $v\in \mathbb{Q}^q$, we denote the number  $\sum_{i=0}^{q-1}v_i\cdot b_i $ it represents in basis $B$ by $[v]_B$.

\subsection{General Algebraic Case}

The existence of a solution of \cref{eq:utility_inequalities} for an algebraic instance of our problem can be expressed as an existential sentence in the language of the reals without exponentiation.
All occurring constants are algebraic numbers of the form $c\cdot b^{\frac{p}{q}}$ with $p, q \in \Naturals$ and $c\in \mathbb{Q}$ in this case.
These can be represented in linear space using exponentiation by repeated squaring.
\begin{lemma}
	Let $b = \frac{u}{v} \in \Rationals$ a rational number where $u,v\in \Naturals\setminus\{0\}$.
	Given $p \in \Integers$ and $q \in \Naturals \setminus \{0\}$, we can define $b^{\frac{p}{q}}$ by an existential sentence in the language of the reals using $\mathcal{O}(\log p + \log q + \log u + \log v)$ symbols.
\end{lemma}
\begin{proof}
	We provide a formula defining $x$ such that $x = b^{\frac{p}{q}}$.
	Let $p$ and $q>0$ be natural numbers with binary representations $p = \sum_{i=0}^P p_i\cdot 2^i$ and $q = \sum_{i=0}^Q q_i\cdot 2^i$ where $p_i, q_i \in \{0,1\}$ for all $0 \leq i \leq P$ and $Q$, respectively.
	The formula $x = b^{\frac{p}{q}}$ with free variable $x$ is equivalent to $x^q = b^p$. Using $d_0=b$ as abbreviation for $v \cdot d_0 = u$, the formula $x^q = b^p$ can then be written as
	\begin{align*}
		\exists c_0, \dots, c_Q, d_0, \dots, d_P. & c_0 = x \land {\lAnd}_{i=1}^Q (c_i = c_{i-1} \cdot c_{i-1}) \land         \\
		                                          & d_0 = b \land {\lAnd}_{i=1}^P (d_i = d_{i-1} \cdot d_{i-1}) \land \\
		                                          & {\prod}_{\{i \mid p_i = 1\}} d_i = {\prod}_{\{i \mid q_i = 1\}} c_i
	\end{align*}
	This formula is uniquely satisfied by assigning $c_i = x^{2^i}$ and $d_i = b^{2^i}$.
	Observe that $q = \sum_{i=0}^Q 2^i q_i$, and consequently $x^q = \prod_{i=0}^Q x^{2^i \cdot q_i} = \prod_{\{i \mid q_i = 1\}} x^{2^i}$, and similar for $b^p$.
	For the size bound, observe that $Q = \lceil \log_2 q \rceil$ and $P = \lceil \log_2 p \rceil$, thus the conjunctions and multiplication require $\mathcal{O}(\log p + \log q)$ symbols.
	Finally, $v \cdot d_0 = u$ requires $\mathcal{O}(\log v + \log u)$ space to specify.
\end{proof}
Consequently, an existential sentence in the language of the reals expressing the satisfiability of \cref{eq:utility_inequalities} can be computed in polynomial time from $\SG$, $b$, and $\gamma$.
This shows \cref{thm:generalalgebraic}.

\subsection{Threshold Problem in Algebraic Extensions of Low Degree}
In the  approach above, we invoke a decision procedure for the full existential theory of the reals.
The problem we solve, however, is of a rather simple form.
In fact, in Markov decision processes, solving the threshold problem boils down to checking the solvability of a linear system of inequalities.

Throughout this section, fix a small algebraic instance of an MDP $\MDP = (\States, \Actions, \sgtransitions)$ and initial state $\initialstate$. 
All coefficients occurring in \cref{eq:utility_inequalities} now are of the form $\sgtransitions(s,a,t) \cdot b^{-r(s) \cdot p/q}$ for a rational probability $\sgtransitions(s,a,t)$, a natural reward $r(s)$, and the rational basis $b$.
Consequently, all coefficients are contained in the field extension $\mathbb{Q}(b^{1/q})$.

We now  determine the degree  and the minimal polynomial of $b^{1/q}$.
Let $b = b_1/b_2$ for co-prime natural numbers $b_1$ and $b_2$.
Now, let $d$ be the greatest divisor of $q$ such that $b_1$ and $b_2$ have integral $d$-th roots $d_1$ and $d_2$.
Then, we can rewrite $b^{1/q}=(d_1/d_2)^{d/q}$.
We now define $q' = q/d$ and $b'=d_1/d_2$.
Thus, $(b')^{1/q'}=b^{1/q}$ and $q'$ is a natural number smaller or equal to $q$.
Now, $q'$ has no divisor $d'$ such that $b'$ has a rational $d'$th root.
By \cite[Thm~8.1.6]{karpilovsky1989topics}, we conclude that $x^{q'}-b'$ is irreducible and hence the minimal polynomial of $b^{1/q}$.

In fact, we can switch to the risk aversion factor $\gamma' = p/q'$ and the basis $b'$ instead of $\gamma$ and $b'$ without affecting the  negative exponential utility.
The entropic risk is changed by the integral factor $d\leq q$ by this switch.
Hence, we from now on work with the following assumption:
\begin{assumption} \label{ass:irreducible}
	In a small algebraic instance with risk aversion factor $\gamma=p/q$ and basis $b$, we can assume w.l.o.g.\ that $x^q-b$ is irreducible over $\mathbb{Q}$.
\end{assumption}
In case of an MDP in which all states belong to the Minimizer, i.e.\ $\Utility$ is maximized, we obtain the following system of linear inequalities from \cref{eq:utility_inequalities}:
\begin{equation} \label{eq:utility_linear_porgram}
	\begin{gathered}
		v(s) = 0 \quad \text{for $s \in S_\infty$,} \qquad v(s) = 1 \quad \text{for $s \in S_0$,} \\
		v(s) \geq b^{-\gamma \reward(s)} \cdot \ExpectedSumSG{\sgtransitions}{s}{a}{v} \quad \text{for $s \in \StatesMin$, $a \in \Actions(s)$,} \\
		v(\initialstate) \leq b^{-\gamma t}.
	\end{gathered}
\end{equation}
Note that we can drop the disjunction from \cref{eq:utility_inequalities} as the point-wise least solution to \cref{eq:utility_linear_porgram} always satisfies this disjunction.
However, if instead all states belong to the Maximizer we cannot drop the disjunction directly.
Nevertheless, we consider the following linear program
\begin{equation} \label{eq:utility_linear_porgram2}
	\begin{gathered}
	\text{Maximize } v(\initialstate) \text{ subject to}\\
		v(s) = 0 \quad \text{for $s \in S_\infty$,} \qquad v(s) = 1 \quad \text{for $s \in S_0$,} \\
		v(s) \leq b^{-\gamma \reward(s)} \cdot \ExpectedSumSG{\sgtransitions}{s}{a}{v} \quad \text{for $s \in \StatesMax$, $a \in \Actions(s)$.} 
	\end{gathered}
\end{equation}
Then, in the optimal solution, we have to check whether $v(\initialstate) \leq b^{-\gamma t}$; or in other words, all solutions to the constraints have to satisfy $v(\initialstate) \leq b^{-\gamma t}$.
By considering the dual linear program instead, we can transform this linear program into one where the objective function has to be minimized and has the same optimal value $v(\initialstate)$.
Here, comparing this optimal value to the threshold by $v(\initialstate) \leq b^{-\gamma t}$ boils down to finding one solution that satisfies all constraints and the threshold condition again.
So, for both Maximizer- and Minimizer-MDPs, the entropic risk threshold problem  boils down to checking the feasibility of a linear program, i.e.\ the satisfiability of a system of linear inequalities.

In \cite{AdlerB94,Beling01}, it is shown that the feasibility of a linear program over algebraic numbers can be checked in time polynomial in the size of the representations of the occurring algebraic numbers and in the degree of a field extension of $\mathbb{Q}$ that contains all coefficients.
The representation of algebraic numbers used in \cite{AdlerB94,Beling01} is the following:
An algebraic number $\alpha$ is represented by its minimal polynomial $P$ over $\mathbb{Q}$ together with a rational interval $(a,b)$ containing only the zero $\alpha$ of $P$.
The size of the representation is the bitsize of coefficients of $P$ and the interval bounds of $a$ and $b$ as fractions of (co-prime) integers.

By \cref{ass:irreducible}, the degree of the field extension containing all occurring numbers is $q$ and hence at most linear in the input size since $q$ is given in unary.
For the following lemma, we show that the coefficients in the minimal polynomials of all constants in the linear program of interest are small as well. Then, the mentioned polynomial-time algorithms for the feasibility of linear systems of inequalities over algebraic numbers from \cite{AdlerB94,Beling01} are applicable.
\begin{restatable}{lemma}{smallalgebraicMDP}\label{lem:small_algebraic_MDP}
	Let $\MDP$ be a small algebraic instance of an MDP.
	We can  decide whether $\ERisk^{\ast}_{\MDP} \geq t$ and whether $\ERisk^{\ast}_{\MDP} \leq t$ in polynomial time.
	Furthermore, for a rational threshold $t^\prime$ given in binary, we can decide whether $\Utility^{\ast}_{\MDP} \leq t^\prime$ and whether $\Utility^{\ast}_{\MDP} \geq t^\prime$ in polynomial time.
\end{restatable}

\begin{proof}
	In order to rely on the mentioned polynomial algorithms 
	from \cite{AdlerB94,Beling01}, we have to show that we can obtain small representations of all coefficients in \cref{eq:utility_linear_porgram,eq:utility_linear_porgram2}.
	The coefficients have the form $\sgtransitions(s,a,s')\cdot b^{-r(s)\cdot p/q}$ where $\sgtransitions(s,a,s')$ and $b$ given in binary and $r(s)$, $p$, and $q$ in unary.
	
	First of all, we note that it is sufficient to find a small representation for $b^{-r(s)\cdot p/q}$.
	As shown in \cite{AdlerB94,Beling01}, from this representation, a representation of $\sgtransitions(s,a,s')\cdot b^{-r(s)\cdot p/q}$ can be found in polynomial time.
	So, let $n=r(s)\cdot p$.
	We rewrite $n=k \cdot q +\ell$ for some natural number $\ell<q$.
	The numerical values of $k$ and $\ell$ are polynomial in the size of the original input as $r(s)$, $p$ and $q$ are given in unary.
	Now, $b^{-n/q}=b^{-k}\cdot b^{-\ell/q}$.
	The number $b^{-k}$ is a rational whose size in binary representation is linear in the size of the binary representation of $b$ and in $k$.
	So, $b^{-k}$ can be obtained in polynomial time and by the same argument as before, it is sufficient to find a small representation for $ b^{-\ell/q}$.

	The algebraic number $b^{-\ell/q}$ is a zero of the polynomial $x^q-b^{-\ell}$.
	Furthermore, it is the only positive real zero of this polynomial and it lies between $0$ and $1$.
	Consequently, the minimal polynomial of $b^{-\ell/q}$ has only one zero in the interval $(0,1)$, too, which yields the interval for the representation.
	In order to find the minimal polynomial, note that $b^{-0/q}$, $b^{-1/q}$, \dots, $b^{-(q-1)/q}$ are linearly independent over $\mathbb{Q}$ by \cref{ass:irreducible}.
	Let $m = q / \gcd(\ell,q)$.
	Now, the remainders of $0,\ell,2\ell,3\ell,\dots$ are periodic modulo $q$ with period $m$.
	This means that $b^{-0\ell/q}$, $b^{-\ell/q}$, \dots, $b^{-(m-1)\ell/q}$ are linearly independent over $\mathbb{Q}$ and consequently the degree of $b^{-\ell/q}$ is at least $m$.
	But $(b^{-\ell/q})^m$ is a rational as $\ell \cdot m$ is an integer multiple of $q$.
	Together, $x^m-(b^{-m\cdot \ell/q})$ is the minimal polynomial of $b^{-\ell/q}$.
	As $\ell\cdot m/q$ is a natural number smaller than $q$, the rational coefficient $b^{-\ell\cdot m/q}$ can be computed in polynomial time.
	Hence, the representation of $b^{-\ell/q}$ in terms of this minimal polynomial and the interval $(0,1)$ can be obtained from $\MDP$ in polynomial time.
	
	As the threshold $t$ on the entropic risk in a small algebraic instance is given in unary, we can similarly obtain a small representation for the corresponding threshold $b^{-\gamma\cdot t}$ 
	on the negative utility. If we directly want to compare the negative utility to a rational threshold $t'$ instead, this threshold can be given in binary.
	Now, the polynomial time algorithms of \cite{AdlerB94,Beling01} for checking the feasibility of linear programs with algebraic coeffiecients are applicable after computing the representations of all constants occurring in \cref{eq:utility_linear_porgram,eq:utility_linear_porgram2} in polynomial time.
\end{proof}

For SGs, this result allows us to non-deterministically guess a strategy for one of the players and check the threshold condition in the resulting MDP in polynomial time.
By guessing a strategy for Maximizer, we can check whether $\ERisk^{\ast}_{\SG}\geq t$ in $\mathsf{NP}$, by guessing a strategy for the Minimizer, we obtain a $\mathsf{coNP}$-upper bound analogously,
finishing the proof of \cref{thm:smallalgebraicthreshold}.
Note that for a rational threshold $t^\prime$ given in binary, the problem to decide whether $\Utility^{\ast}_{\SG} \geq t^\prime$ for a small algebraic instance $\SG$ of an SG is in $\mathsf{NP}\cap\mathsf{coNP}$ as well with analogous reasoning and \cref{lem:small_algebraic_MDP}.

Finally, we show that despite the restrictions on the encoding of the input our problem 
is still at least as hard as for general reachability games:
\begin{lemma}
	The threshold problem for stochastic reachability games is polynomial-time reducible to the entropic risk threshold problem on small algebraic instances of SGs.
\end{lemma}

\begin{proof}
	This follows  from \cref{lem:capturing_reachability}:
	For each stochastic reachability game $\SG$, we can construct an SG $\SG'$ with rewards $0$ and $1$, basis $b=2$, and $\gamma=1$ in polynomial time such that the reachability value in $\SG$ equals the optimal entropic risk in $\SG'$.
	Clearly, this is a small algebraic instance.
\end{proof}

\subsection{Optimal Solution in Algebraic Extensions of Low Degree}
\label{sec:explicit}

So far, we presented solutions to the threshold problem.
We will now address how to compute an explicit representation of the optimal negative exponential utility.
To this end, we take a closer look at how the equation system we obtain for MCs can be solved with computations in the algebraic numbers.
In the end, we state the complexity-theoretic consequences for SGs that we obtain from this solution for MCs.

Let $\MC = (\States, \mctransitions)$ be a small algebraic instance of a MC with reward function $r : S \to \mathbb{N}$, initial state $\initialstate$, $\gamma = p/q \in \mathbb{Q}$ a risk-aversion parameter for co-prime integers $p$ and $q$, and $b\in \mathbb{Q}$ a basis.
Again, we work with \cref{ass:irreducible} as justified in the previous section.
The main result of this section is the following lemma leading to \cref{thm:algebraic_optimal_value}.
\begin{lemma} \label{thm:algebraic_MC}
	A representation of the negative exponential utility $\Utility^*$ in $\MC$ in the basis $B=(1,b^{1/q},b^{2/q},\dots, b^{q-1/q})$ of $\mathbb{Q}(b^{1/q})$ can be computed in polynomial time.
\end{lemma}
First, we compute sets $S_0$ and $S_\infty$ of states from which reward $0$ or $\infty$, respectively, is received almost surely.
This can be done by the analysis of bottom strongly connected components and the computation of reachability probabilities in MCs in polynomial time, but of course it is also a special case of the analogous computation of such sets in games above.
The system of inequalities \cref{eq:utility_inequalities} simplifies to 
\begin{equation} \label{eq:utility_MC}
	\begin{gathered}
		v(s) = 0 \quad \text{for $s \in S_\infty$,} \qquad v(s) = 1 \quad \text{for $s \in S_0$,} \\
		v(s) = {\sum}_{t\in S} b^{-\gamma\cdot r(s)} \cdot \mctransitions(s,t) \cdot v(t) \quad \text{for $s \in S\setminus(S_0\cup S_{\infty})$.} 
	\end{gathered}
\end{equation}
This is a linear equation system $Ax=v$ where the entries of $A$ and $v$ are of the form $c \cdot b^{-\gamma w}$ for rational $c$ and natural numbers $w$.
Due to the pre-processing of $S_0$ and $S_{\infty}$, we know that $Ax=v$ has a unique solution and that this solution contains the negative exponential utility from each state.

As we assume that all rewards are natural numbers, all entries in $A$ and $v$ are elements of the number field $\mathbb{Q}(b^{1/q})$.
By \cref{ass:irreducible},
 $\mathbb{Q}(b^{1/q})$ is a $q$-dimensional $\mathbb{Q}$-vector space.
The tuple $B=(1,b^{1/q},b^{2/q},\dots, b^{(q-1)/q})$ forms a basis of $\mathbb{Q}(b^{1/q})$.
So, each element  $a \in \mathbb{Q}(b^{1/q})$ can be represented by a vector $z \in \mathbb{Q}^q$ with $[z]_B = \sum_{i=0}^{q-1} z_{i} \cdot b^{i/q}=a$.

Now, our goal is to compute an explicit representation of the value $v(\initialstate)$ in the unique solution of \cref{eq:utility_MC}, i.e.\ a vector $z \in \mathbb{Q}^q$ such that $[z]_B = v(\initialstate)$.
For this purpose, we solve \cref{eq:utility_MC} via Gaussian elimination by explicit computations in $\mathbb{Q}(b^{1/q})$ using representations in the basis $B$.
The following results lead to the main result of this section. 
\begin{restatable}{lemma}{representationcoefficients} \label{lem:writing}
	For a number of the form $c \cdot b^{-n/q}$ with $c \in \mathbb{Q}$ and $n \in \mathbb{N}$ we can compute its representation in basis $B$ in time polynomial in the length of the binary representation of $c$, and polynomial in the numerical values of $n$ and $q$.
\end{restatable}

\begin{proof}
	We can write $-n/q$ as $-k+\ell/q$ for natural numbers $k$ and $\ell$ with $\ell<q$.
	Then, $c \cdot b^{-n/q} = c \cdot b^{-k} \cdot b^{\ell/q}$ which has the representation $c \cdot b^{-k}\cdot e_\ell$ in $\mathbb{Q}^q$ where $e_\ell$ is the $\ell$-th standard basis vector.
	The coefficient $c \cdot b^{-k}$ can be computed in time polynomial in the representation of $c$ and in the numerical value of $k$ which is bounded by the numerical value of $n$.
\end{proof}

\begin{restatable}{lemma}{multiplication} \label{lem:multiplication}
	Given two representations $z, z' \in \mathbb{Q}^q$ in basis $B$ of numbers $y=[z]_B $ and $y' = [z']_B $ in $\mathbb{Q}(b^{1/q})$, we can compute the representation of $y \cdot y'$ in time polynomial in the size of the binary representations of $z$ and $z'$.
	Furthermore, we can compute the representation of $y^{-1}$ in time polynomial in the size of the binary representations of $z$.
\end{restatable}

\begin{proof}
	Multiplication is a bi-linear map.
	Thus, in order to determine a representation of $[z]_B \cdot [z']_B$ for given $z, z' \in \mathbb{Q}^q$, it is sufficient to know the $q^2$-many representations $m_{\ell, h}$ of $[e_{\ell}]_B \cdot [e_h]_B$ for standard basis vectors $e_{\ell}$ and $e_h$ with $0 \leq \ell,h \leq q-1$.
	The representation of $[z]_B \cdot [z']_B $ can then be computed as the sum
	\begin{equation} \label{eqn:multilpication}
		{\sum}_{0\leq \ell,h\leq q-1} z_{\ell} \cdot z'_h\cdot m_{\ell,h}.
	\end{equation}
	where the vectors $m_{\ell, h}$ are the representations of $b^{(\ell-1)/q}\cdot b^{h/q}= b^{(\ell+h)/q}$.
	If $\ell+h<q$, the representation is $m_{\ell, h} = e_{\ell+h}$.
	If $q \leq \ell+h < 2q$, the value is $b \cdot b^{\ell+h-q}$ and its representation is $b \cdot e_{\ell+h-q}$.
	This proves the first claim.

	In order to compute the representation of the inverse $y^{-1}$, we take a vector of variables $x$.
	The equation
	\begin{equation*}
		{\sum}_{0\leq \ell,h \leq q-1} z_{\ell} \cdot x_h \cdot m_{\ell,h} = e_0
	\end{equation*}
	yields a rational equation system with $q$ variables which has a unique solution as the inverse is unique.
	Given $z$, the system can be constructed in polynomial time using the $q^2$-many vectors $m_{\ell,h}$.
	Consequently, the representation in basis $B$ of $y^{-1}$ can be computed in time polynomial in the binary representation of $z$, proving the second claim.
\end{proof}

\begin{restatable}{lemma}{gaussian}
	We can perform Gaussian elimination on $Ax=v$ to obtain a representation of $x$ in the basis $B$ in time polynomial in the encoding size of the given small algebraic instance of a MC.
\end{restatable}

\begin{proof}
	We perform Gaussian elimination on the matrix $A$ and vector $v$ in a way that ensures that all intermediate numbers have a small representation.
	In $A$ and $v$, we write representations from $\mathbb{Q}^q$ at every entry.
	Writing down the system with representations can be done in polynomial time by \cref{lem:writing}.
	In \cite{edmonds1967systems}, it is shown that it is possible to perform Gaussian elimination in a way such that all numbers occurring during the computation are determinants of a submatrix of the original input.
	If the problem dimension is $n$, this means that all numbers occurring during the computation are sums of $n!$ products of at most $n$ numbers from the original input.
	In our case, $n$ is the number of states of the Markov chain $M$.

	Let $d$ be the least common multiple of all denominators of numbers in $A$ and $v$.
	The bit size of $d$ is linear in the bit sizes of the denominators.
	Generalizing \cref{eqn:multilpication} in the proof of \cref{lem:multiplication}, we obtain that the product of $n$ numbers given as representations $x^i\in \mathbb{Q}^q$ for $1\leq i \leq n$ is
	\begin{equation*}
		y = \sum_{\substack{0\leq \ell_i \leq q-1 \\ 1\leq i \leq n}} \prod_{i=1}^{n} (x^{i})_{\ell_i} \cdot b^{\lfloor \sum_{i=1}^{n} (x^{i})_{\ell_i}/q \rfloor} \cdot e_{(\sum_{i=1}^{n} (x^{i})_{\ell_i} \mathop{\mathsf{mod}} q)}.
	\end{equation*}
	If $m$ is the maximal absolute value of any entry in one of the $x^i \in \mathbb{Q}^q$ for $1 \leq i \leq n$, each component of this vector $y$ has absolute value less than $q^n \cdot m^n \cdot b^n$.
	Furthermore, if $b=b_1/b_2$, each component of $y$ is an integer multiple of $1/(b_2^n\cdot d)$.
	So, each component $y_i$ is the fraction of an integer less than $q^n \cdot m^n\cdot b^n \cdot b_2^n \cdot d$ and an integer less than $b_2^n\cdot d$.
	The bit size of these numbers is bounded by $n \cdot (\log_2(q)+\log_2(m)+\log_2(b)+\log_2(b_2)) + \log_2(d)$.
	Thus, the bit size of all components of $y$ is at most polynomial in the bit sizes of the entries of the vectors $x^i\in \mathbb{Q}^q$ for $1\leq i \leq n$.

	Together, any product of $n$ entries of $A$ and $v$ has a representation in basis $B$ of polynomial size.
	Furthermore, all numbers that can occur are integer multiples of $1/(b_2^n\cdot d)$.
	Consequently, if we want to add $n!$-many such numbers, we can rewrite all rational numbers to denominator $1/(b_2^n\cdot d)$ and afterwards add the integer enumerators component-wise.
	This increases the bit size by a factor of at most $\log_2(n!) < \log_2(n) \cdot n$.
	So, all intermediate numbers that occur when performing Gaussian elimination as in \cite{edmonds1967systems} have a representation whose bit size is bounded by a polynomial in the input size.
	As $Ax=v$ has a unique solution, the Gaussian elimination produces this solution in polynomially many steps and all necessary multiplications and divisions can be carried out in polynomial time by \cref{lem:multiplication} and the fact that all intermediate numbers occurring have a polynomially large representation.
\end{proof}
Put together, this finishes the proof of \cref{thm:algebraic_MC}.
So, on small algebraic instances of MCs, we can compute and explicit representation of the negative exponential utility in polynomial time.
This allows us to conclude that
	on a small algebraic instance of an SG,  an explicit representation of $\ERisk^*$ can be computed in polynomial space:
	In polynomial space, we can go through all MD-strategies $\sigma$ for the Maximizer.
	For each strategy $\stratmax$, we compute a representation of $\Utility$  for each Minimizer MD-strategy $\stratmin$ in the resulting MC and compare it to the least value we have seen so far that the Minimizer can enforce against $\stratmax$.
	To compare the explicit representations computed in the process, we can rely on the algorithms in \cite{AdlerB94,Beling01} to compare algebraic numbers.
	Once we have found the value the Maximizer can enforce with $\stratmax$, we consider the next strategy of the Maximizer and keep track of the best value found so far.
This concludes the proof of
\cref{thm:algebraic_optimal_value}.

\subsection{Integer Exponents}

To conclude, we consider the special case that $b$ is a rational number and $\gamma \cdot \reward(s)$ is an integer for all states $s$.
Then, the transition probabilities $b^{-\gamma \reward(s)}$ are rational.
In this case, we can directly compute $\SG_R$, requiring space linear in the numerical values $\gamma \cdot \reward(s)$.
Then, we can apply standard methods to decide whether the Maximizer can ensure a reachability probability of at most $b^{-\gamma t}$.
In particular, for a fixed upper bound on $\gamma \cdot \reward(s)$, this yields a polynomial procedure for MCs and decision processes.
Note that even if $\gamma t$ is not an integer, we could compute the optimal reachability precisely and then check whether this (rational) value is larger or smaller than the threshold by computing sufficiently many digits of $b^{-\gamma t}$.

\section{Approximation} \label{app:approximation}

More formally, let $\SG_R$ the reachability game from \cref{stm:utility_reachability_equivalence}.
We define a new game $\SG_\approx$, only multiplicatively changing the transition probabilities of $\SG_R$.
In particular, we have $\mdptransitions^\approx(s, a, t) = \mdptransitions^R(s, a, t) \cdot (1 + \delta_s) \in \Rationals$ and $\abs{\delta_s}$ is small.
By slightly modifying statements of \cite{chatterjee2012robustness}, we obtain a general bound on the relative difference between reachability values in $\SG_R$ and $\SG_{\approx}$.
In the following, we require some tools of \cite{chatterjee2012robustness}.
The exact definitions are rather technical, we refer to \cref{app:robustness_definitions} for a summary and \cite{chatterjee2012robustness} for a more complete picture.
In a nutshell, we use the notion of \emph{structurally equivalent games}.
Intuitively, this means that the states, actions, and supports of transitions are equivalent; in other words, the induced graphs are the same.
For such equivalent games $\SG_1$ and $\SG_2$ one can define the \emph{relative difference} $\mathrm{dist}_R(\SG_1, \SG_2)$, which refers to the largest quotient of the probabilities of two corresponding transitions minus 1.
In turn, this distance bounds the difference in reachability values.

Intuitively, for each state $s$ we want to find a factor $\delta_s$ such that (i)~$b^{-\gamma \reward(s)} \cdot (1 + \delta_s)$ is rational, (ii)~has a sufficiently small denominator to be computationally viable, and (iii)~not change the obtained value too much.

\begin{restatable}{lemma}{restatestructurally} \label{stm:game_rounding}
	Let $\SG_1$, $\SG_2$ two structurally equivalent games together with a reachability objective $T$.
	Set $0 \leq d = \mathrm{dist}_R(\SG_1, \SG_2)$.
	Then
	\begin{equation*}
		(1 + d)^{-2 \cardinality{S}} \leq \frac{\val(\SG_1, T)}{\val(\SG_2, T)} \leq (1 + d)^{2 \cardinality{\States}}.
	\end{equation*}
\end{restatable}

\begin{proof}
	We modify proofs of \cite{chatterjee2012robustness} as follows.

	First, consider \cite[Lem.~3]{chatterjee2012robustness}:
	During the proof, we get that
	\begin{equation*}
		(1 + d)^{- 2 \cardinality{S}} \leq \frac{\val(\MC_1, \mathsf{MDT}(\lambda, r))(s)}{\val(\MC_2, \mathsf{MDT}(\lambda, r))(s)} \leq (1 + d)^{2 \cardinality{S}},
	\end{equation*}
	where $\MC_1$ and $\MC_2$ are two structurally equivalent Markov chains, $d = \mathrm{dist}_R(\MC_1, \MC_2)$ their relative difference, and $\mathsf{MDT}$ is the \emph{multi-discounted objective function}, a technical device of \cite{chatterjee2012robustness}.
	See \cref{app:robustness_definitions} for further details.

	Continuing with \cite[Thm.~4]{chatterjee2012robustness}, we obtain through \cite[Thm.~2]{chatterjee2012robustness} that the same inequality also holds for the value of parity objectives and, as a special case, for reachability.
	By then applying the reasoning of \cite[Thm.~5]{chatterjee2012robustness}, i.e.\ considering the Markov chain obtained by fixing two optimal memoryless deterministic strategies, the inequality transfers to games.
\end{proof}

It remains to show that we can indeed obtain reasonably small rational transition probabilities through this rounding.
\begin{restatable}{lemma}{restaterounding} \label{stm:rounding}
	Fix a precision requirement $\varepsilon > 0$, let $\reward_{\min}$ and $\reward_{\max}$ equal the minimal and maximal occurring non-zero rewards, respectively, $N = \cardinality{\States}$ the number of states, and $p_{\min}$ the smallest occurring non-zero transition probability.

	Then, there exists a rounded game $\SG_{\approx}$ such that (i)~the reachability probability in $\SG_{\approx}$ relatively differs from $\SG_R$ by at most $b^{\gamma \varepsilon}$ and (ii)~all transition probabilities are rational quantities with a denominator of bit size
	\begin{gather*}
		- \log_2 p_{\min} - \min(\log_2 b^{-\gamma \reward_{\max}}, \log_2 (1 - b^{-\gamma \reward_{\min}})) \\
		- \log_2 \gamma - \log_2 \varepsilon + \log_2{N} - \log_2{\log b}.
	\end{gather*}
\end{restatable}

\begin{proof}
	We want to show that
	\begin{equation*}
		b^{-\gamma \varepsilon} \leq \val(\SG_R, T) / \val(\SG_{\approx}, T) \leq b^{\gamma \varepsilon}.
	\end{equation*}
	By \cref{stm:game_rounding}, this holds if the transition probabilities have a relative difference $d$ with
	\begin{equation*}
		b^{-\gamma \varepsilon} \leq (1 + d)^{-2 \cardinality{S}} \quad \text{ and } \quad (1 + d)^{2 \cardinality{S}} \leq b^{\gamma \varepsilon}.
	\end{equation*}
	Or, rearranged, $1 + d \leq b^{\gamma \varepsilon / (2 \cardinality{\States})}$ (and $0 \leq d$).
	To ease notation, we define $z = \gamma \varepsilon / (2 \cardinality{\States})$.

	Fix some state $s \in \States$.
	Recall that the reachability game $\SG_R$ features two kinds of transition probabilities in each state.
	First, transition probabilities of the original game multiplied by $b^{-\gamma \reward(s)}$, and second the transition to the introduced trap state, multiplying by $(1 - b^{-\gamma \reward(s)})$.

	Let us focus on the first kind and suppose the transition probability is given by $b^{-\gamma \reward(s)} \cdot p$ and assume that $\reward(s) \neq 0$.
	We want to show that there exists a rational number with a sufficiently small representation in the neighbourhood of this transition probability, i.e.\ in the interval $I = [b^{-\gamma \reward(s)} p \cdot b^{-z}, b^{-\gamma \reward(s)} p \cdot b^z]$.
	Such a number necessarily exists if this interval is sufficiently large.
	Thus, consider $b^z - b^{-z}$.
	By change of base, we obtain
	\begin{equation*}
		b^z - b^{-z} = e^{z \cdot \log b} - e^{-z \cdot \log b} = 2 \sinh(z \log b) \geq 2 z \log b,
	\end{equation*}
	using that $\sinh(x) \geq x$.
	Thus, the interval has at least size $\cardinality{I} \geq b^{-\gamma \reward(s)} p \cdot 2 z \log b$.
	Now, if $n$ satisfies $2^{-n} \leq \cardinality{I}$, then for some $m$ we have that $\frac{m}{2^n} \in I$.
	In other words, $n$ is an upper bound on the bit size of the smallest denominator that can be found in that interval.
	Taking $\log_2$ and inserting $z$, we arrive at
	\begin{align*}
		n \geq & - \log_2 b^{-\gamma \reward(s)} - \log_2 \gamma - \log_2 \varepsilon \\
		& + \log_2 p + \log_2 \cardinality{\States} + 1 - \log_2 \log b.
	\end{align*}
	Thus, $I$ contains a rational number with a denominator of that bit size.
	For the other type of transition, observe that now the interval of concern is centred around $1 - b^{-\gamma \reward}$ and we analogously obtain
	\begin{align*}
		n \geq & - \log_2 (1 - b^{-\gamma \reward(s)}) - \log_2 \gamma - \log_2 \varepsilon \\
		& + \log_2 p + \log_2 \cardinality{\States} + 1 - \log_2 \log b.
	\end{align*}

	Recall that we assumed $\reward(s) \neq 0$.
	If instead we have $\reward(s) = 0$ and $p$ is rational, we are done, since the transition probabilities from this state in $\SG_R$ are equal to those of the original game $\SG$, i.e.\ rational and of size $-\log_2 p$.
	In case $p \in \Reals \setminus \Rationals$, observe that we only need to round the first kind of transitions (since the second kind is zero).
	However, here we can apply the same reasoning and analogously obtain the first inequality, noting that $\log b^{-\gamma \reward(s)} = 0$ in this case.

	Taking the maximum of all inequalities over all states yields the result.
\end{proof}

For practical application, note that we do not (and cannot) compute  the concrete factor $b^{-\gamma \reward(s)}$ or the rounding $\delta_s$.
Instead, we determine as many digits of $b^{-\gamma \reward(s)}$ as required by the inequality in the proof.
This guarantees that the obtained value is within the required interval $I$.
For \enquote{reasonable} encodings of $b$, $\gamma$, and $\reward(s)$, these digits can be efficiently obtained through standard means.

\subsection{Definitions for Robustness} \label{app:robustness_definitions}

In this section, we briefly repeat the definitions of \cite{chatterjee2012robustness} which are necessary for the proof of \cref{stm:game_rounding}.
We note that the results of \cite{chatterjee2012robustness} apply to \emph{concurrent} stochastic games with \emph{parity} objective, which are a generalization of \emph{turn-based} stochastic games with \emph{reachability} objective.
We rephrase the definitions relative to our model.

Fix two stochastic games $\SG_1 = (\StatesMax^1, \StatesMin^1, \Actions_1, \sgtransitions_1)$ and $\SG_2 = (\StatesMax^2, \StatesMin^2, \Actions_2, \sgtransitions_2)$
We say that two games  and are \emph{structurally equivalent} if they induce the same graph, formally $\StatesMax^1 = \StatesMax^2$, $\StatesMin^1 = \StatesMin^2$, $\Actions_1 = \Actions_2$, and $\support(\sgtransitions_1(s, a)) = \support(\sgtransitions_2(s, a))$ for all $s \in \StatesMax^1 \union \StatesMin^1$ and $a \in \Actions_1(s)$.
Since the set of states and actions is equal, we omit the subscripts in the following.
Moreover, let $\States = \StatesMax^1 \union \StatesMin^2 = \StatesMax^2 \union \StatesMin^2$.

We define the \emph{distance} between two structurally equivalent games as
\begin{multline*}
	\mathrm{dist}_R(\SG_1, \SG_2) \coloneqq \max \bigg\lbrace \frac{\sgtransitions_1(s, a, t)}{\sgtransitions_2(s, a, t)}, \frac{\sgtransitions_2(s, a, t)}{\sgtransitions_1(s, a, t)} \mid \\ s \in \States, a \in \Actions(s), t \in \support(\sgtransitions_1(s, a, t)) \bigg\rbrace - 1.
\end{multline*}
Since the games are structurally equivalent, the fractions are always well defined.
Moreover, the value is always $\geq 0$ and $= 0$ if and only if the two games are equal.
Intuitively, the distance is the largest relative difference between two corresponding transition probabilities minus 1.

Next, the \emph{mean-discounted time} for a state $s$, discount vector $\lambda : \States \to \Reals$, and infinite path $\infinitepath$ refers to the discounted time the path is in that state, formally
\begin{equation*}
	\mathsf{MDT}(\lambda, s)(\infinitepath) \coloneqq \frac{\sum_{j=0}^\infty (\prod_{i=0}^j \lambda(\infinitepath_i)) \cdot \indicator{s}(\infinitepath_j)}{\sum_{j=0}^\infty (\prod_{i=0}^j \lambda(\infinitepath_i))}.
\end{equation*}

The starting point to proving the result in \cite{chatterjee2012robustness} then is to show that for Markov chains the expected mean-discounted time can be expressed by a rational function comprising polynomials of bounded degree.
In consequence, it is shown that for two Markov chains which are \enquote{close} w.r.t.\ $\mathrm{dist}_R$ the difference between the expected mean-discounted time can be bounded, too.
The result follows by the known result that reachability (and parity) can be obtained as limit of mean-discounted time by taking the values of $\lambda$ to $1$.

\end{document}